\documentclass[twoside,fleqn]{ActaStyle}
\usepackage[dvips]{graphicx}
\usepackage{amsthm}
\usepackage{amsmath}
\usepackage{amssymb}
\def\authorheading{V. Scarani}
\def\shorttitle{The device-independent outlook on quantum physics}

\newcommand{\ba}{\begin{eqnarray}}
\newcommand{\ea}{\end{eqnarray}}
\newcommand{\ban}{\begin{eqnarray*}}
\newcommand{\ean}{\end{eqnarray*}}
\newcommand{\braket}[2]{\mbox{$ \langle #1 | #2 \rangle $}}
\newcommand{\moy}[1]{\langle #1 \rangle}
\newcommand{\sandwich}[3]{\mbox{$ \langle #1 | #2 | #3 \rangle $}}

\newcommand{\demi}{\frac{1}{2}}
\newcommand{\compl}{\mathbb{C}}
\newcommand{\real}{\mathbb{R}}

\newcommand{\nat}{\mathbb{N}}
\newcommand{\one}{\leavevmode\hbox{\small1\normalsize\kern-.33em1}}
\newcommand{\Tr}{\mathrm{Tr}}

\newtheorem{theorem}{Theorem}[section]
\newtheorem{lemma}{Lemma}[section]
\newtheorem{proposition}{Proposition}[section]
\newtheorem{corollary}{Corollary}[section]

\input{Qcircuit}

\begin{document}
\pagerange{1}{60}   
\title{The device-independent outlook on quantum physics}
\author{V.~Scarani\email{physv@nus.edu.sg}$^*$}
{$^*$ Centre for Quantum Technologies \& Department of Physics, National University of Singapore, Singapore}
\abstract{This text is an introduction to an operational outlook on Bell inequalities, which has been very fruitful in the past few years. It has lead to the recognition that Bell tests have their own place in applied quantum technologies, because they quantify non-classicality in a device-independent way, that is, without any need to describe the degrees of freedom under study and the measurements that are performed. At the more fundamental level, the same device-independent outlook has allowed the falsification of several other alternative models that could hope to reproduce the observed statistics while keeping some classical features that quantum theory denies; and it has shed new light on the long-standing quest for deriving quantum theory from physical principles.}


\tableofcontents 
\section{Preamble}

\subsection{A parable...}

A physicist turned bureaucrat was sent to audit the quantum laboratories of a university. He entered the first lab and found a student intent in his experiment.\\
-- ``Good morning. I was told that you do exciting physics here. What are you working on?"\\
-- ``This is a Stern-Gerlach experiment. I send a beam of atoms through this inhomogeneous magnetic field and measure the spin, a purely quantum effect."\\
-- ``Nice! This was advanced material in my undergraduate days. You needed that cumbersome Dirac notation. Can you remind me what one looks for in such an experiment?"\\
-- ``The theory says that, if the spin points in direction $\hat{n}$ and you orient the magnet in a direction $\hat{a}$, the probability of finding the atom in the upper beam is $P(+\hat{a}|\hat{n})=\demi(1+\hat{a}\cdot\hat{n})$. I am checking these predictions and it works really very well."\\
The student showed his skills and those of his machine by running the experiment. The match with the prediction was indeed remarkable. The bureaucrat was impressed and asked:\\
-- ``Can you remind me how we conclude that this is a quantum effect?"\\
-- ``Well... you see, the result of the measurement is discrete..."\\
-- ``But so it would be if I would toss a coin and observe head or tail."\\
-- ``Right, but the spin measurement is random. Take your example: the coin looks random because we don't control all the parameters; but here, the outcome is \textit{really} random. We can't predict individual events."\\
The bureaucrat paused to think, and after a while said:\\
-- ``I am not convinced. May I try something? --- No, I won't touch your delicate apparatus. I plan even to go out of this room. I shall simulate a source of spins pointing in the $\hat{z}$ direction. You stay inside, write the measurement direction on a piece of paper and pass it to me below the door. I'll try to reply as the spins would."\\
The student watched the bureaucrat close the door. Then took a piece of paper, wrote $\hat{a}=\frac{1}{2}\hat{z}+\frac{\sqrt{3}}{2}\hat{x}$ on it and slipped it under the door. Soon afterwards, a small paper square flew in: on it was written \textit{up}; then another, and another, and another... up, down, down, up, down, up, up, ... building up the expected statistics $P(up)=75\%$ and $P(down)=25\%$. After a while, the bureaucrat came back in.\\
-- ``So, am I a quantum source? Really random?" he said with a wry smile.\\
The student was flummoxed. He had been told that the Stern-Gerlach magnet produced real randomness. It came back to his mind how Feynman and Schwinger, in their textbooks, use precisely that experiment to build up quantum theory itself. Seeing his idols in danger of being brought down, he became mildly aggressive:\\
-- ``Your tricky game is unfair. It just proves that quantum theory is predictive, that we know what to expect..."\\ The bureaucrat did not wait for the end and went away, stunned by his smartness. The student was left to ponder the end of his own sentence. Indeed: we know what to expect! If this information exists (or can be easily computed by a bureaucrat), why have all been told to believe that the source of atoms does not possess that information in advance?\\

After a well deserved lunch, the bureaucrat went back to his task and knocked at the door of the second lab. A girl emerged from obscurity and politely asked what was the matter. The bureaucrat introduced his auditing role and went straight to the point:\\
-- ``Are you also doing some fundamental experiment here?"\\
-- ``Yes, sir. I am observing the violation of Bell inequalities by two entangled photons."\\
-- ``And you are surely convinced that this is really quantum, are you not? Show me the statistics you are getting and I shall show you something funny."\\
-- ``Ah, sir, you are certainly the one who gave trouble to my friend this morning. He told me what you did, it was intriguing indeed. But you can't do the same here, it would be unfair."\\
-- ``Unfair? Why? Because you want to keep your blind faith in what you were taught? This is science, not -"\\
-- ``That is not the reason, sir. You see, here we have two photons, and each is measured independently at a different location. Each photon cannot possibly know which measurement we perform on its twin, because we choose them at the last moment and the information does not have time to reach the other location."\\
-- ``And so..."\\
-- ``And so, if your simulation has to be fair, I can't give the information about both measurements to you alone. You have to come with a friend. Each of you will have to go to a different room, and without your mobile telephones. Then, I give one measurement to you and another to your friend."\\
-- ``This way, I simulate one photon and my friend the other, right? Fair enough. Now, who -"\\
A voice came from the corridor:\\
-- ``Hey, what a surprise? What are you doing here?"\\
It was a former classmate of the bureaucrat, who had persevered in the academic world and had specialized in mathematical physics. He immediately volunteered to play the friend in the simulation. They sat down together, came up with a simple strategy and underwent the test --- but, when they checked their simulation with the student, they found they had failed to reproduce the expected statistics for some of the measurements. Intrigued by the challenge, they went to the cafeteria and prepared a more elaborated common strategy while sipping a good coffee. The bureaucrat asked:\\
-- ``By the way, what is this violation of Bell inequalities the girl was speaking about? There is a great fuzz about it in some blogs that I follow."\\
-- ``Oh, it's just an obvious consequence of describing two opposite spins $\demi$ in a $C^*$-algebra. Finite-dimensional, non-relativistic quantum mechanics -- as boring as it gets."\\
They went back to the lab and tried again the simulation. It went pretty well, most of the statistics were indeed correct; but not all, and the simulation of some pairs of measurements was wide off the mark. However, the time had come at which bureaucrats are asked to stop working. He took leave from the student and from his friend and walked out in the fresh evening. He had had a very pleasant day out of office: he would recommend the funding for both projects to be renewed.

\subsection{... and its meaning}

The bureaucrat and his friend could have tried much longer: they would \textit{not} have succeeded in simulating the statistics observed in the experiment with entangled particles. \textit{This statement is the operational meaning of the violation of Bell inequalities}. This text will illustrate how \textit{exciting physics arises from this clean operational approach to Bell inequalities}.

It may be superfluous to pass explicit judgment on the characters of the story, insofar as they were playing each their natural role. Let me do it anyway, for the sake of those younger readers who may not be acquainted with the academic world:
\begin{itemize}
\item The mathematical physicist got it wrong. The violation of Bell inequalities is not a straightforward exercise in finite-dimensional quantum mechanics. It is a criterion independent of quantum physics. Finite-dimensional quantum mechanics provides excellent agreement with the observation, which should lead to appreciation rather than dismissal.
\item The bloggers won't be able to separate light and darkness from the primordial chaos in which they are immersed. To be fair, there may be something exciting at the philosophical level about the violation of Bell inequalities. However, we won't need to look outside physics to appreciate the power of Bell.
\item The students are in the process of getting it right: they just have to undergo the critical phase transition when they stop believing blindly in their supervisors and switch their own brains on (not to antagonise the supervisor, but to appreciate the supervisor's wisdom while building up an independent wisdom of their own).
\item The bureaucrat got it absolutely right in deciding to continue funding research.
\end{itemize}

\subsection{A user's guide to this text}

This text is an editing of the lecture notes of a module I am teaching within CQT's graduate programme. It is \textit{not} a review paper: I recently co-authored one such paper \cite{ourreview} and the reader is encouraged to refer to it for a comprehensive overview of the research in this field. Here, I have done a clear selection of material, focused on simple examples and kept the references to a minimum.

The readers are supposed to be familiar with quantum formalism. If this were not the case, they can refer to any of the excellent books and courses which have gone a great didactical length to introduce those notions carefully. For instance, \cite{preskill} are in open access.

\section{Bell inequalities as an operational notion}
\label{secbell}

\subsection{Introductory matters}

\subsubsection{Bell experiments}

This text deals with the description of a \textit{very specific class of experiments} sketched in Fig.~\ref{fig5bellexp}. Two parties Alice and Bob are at distinct locations. Each has a measurement device, which shall be treated as a \textit{black box} with an input (say, a knob, to choose the measurement setting) and an output (to record the result). In each run of the experiment, each party sets the knob at a randomly chosen position and receives an outcome. After repeating the procedure several times, Alice and Bob come together (or exchange information via communication) and compute the joint statistics of their observations\footnote{In this text, Alice and Bob are always the verifiers that operate the black boxes: their role is the one described here, namely to choose measurement settings, record outcomes and compute statistics. In some papers, Alice and Bob are rather those who receive the settings from a referee and are supposed to produce outcomes according to the desired statistics (think of an experimentalist in each box, who has control over the internal mechanisms). In this paper, only in paragraph \ref{sshowmuch} will it be convenient to give names to the simulators, and I shall use Anthony and Beatrix for that.}.

\begin{figure}[ht]
\begin{center}
\includegraphics[scale=0.60]{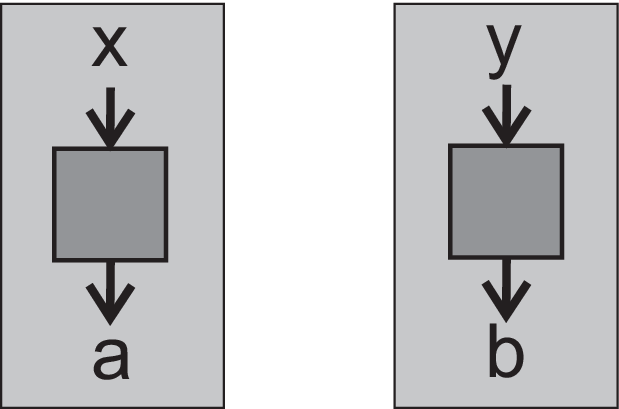}
\caption{Bipartite Bell experiment. Notice that we do not need to specify a ``channel" between the locations: the boxes may have been pre-loaded with shared information (classical or quantum), but in the Bell test Alice and Bob act in a completely independent way.}
\label{fig5bellexp}
\end{center}
\end{figure}

Such experiments have been repeatedly realized in the last few decades but were first proposed as \textit{Gedankenexperimente}. The famous 1935 paper by Einstein, Podolski and Rosen \cite{epr} is obviously a precursor, but the setup they presented would not allow to have two spatially separated measurement stations (see Appendix \ref{aepr}) which, as we shall see, is crucial in our present understanding. This became possible only when David Bohm in 1951 rephrased the EPR argument using entanglement of internal degrees of freedom. It is essentially on Bohm's description that most of the subsequent work relies, including certainly the work of John Bell that eventually put the discussion on the right grounds \cite{bell64}. Nowadays, this class of experiment is referred to as EPR (or EPR-like) experiments, Bell-Bohm experiments, or simply Bell experiments; I shall use the latter.

Let us fix the notation: the possible measurement setting for Alice will be denoted $x\in{\cal X}=\{1,...,M_A\}$, her outcome $a\in{\cal A}=\{1,...,m_A\}$. Similarly, Bob's setting is $y\in{\cal Y}=\{1,...,M_B\}$ and his outcome is $b\in{\cal B}=\{1,...,m_B\}$. Two points are worth stressing:
\begin{itemize}
\item Alice and Bob are not requested to know what each input (each position of the knob) corresponds to inside the device. For all they know \textit{a priori}, it is even possible that all the positions of the knob correspond to the same physical operation inside the box; of course, if this is the case, \textit{a posteriori} they will observe that their joint statistics do not change with the input.
\item The fact that the outcomes are discrete does not imply that quantum degrees of freedom, let alone finite-dimensional ones, are being measured: the boxes may be measuring the frequency of light beams with devices that group the results according to the traditional seven colors of the rainbow.
\end{itemize}
For simplicity of the presentation, let us temporarily assume that \textit{the behavior of the black boxes is the same in each run} (I'll show how to remove this assumption in paragraph \ref{ssgill}). Alice and Bob are sampling from the family of $M_A\times M_B$ probability distributions\footnote{All along this text, I shall use the notation of conditional probabilities $P(a,b|x,y)$, to be read ``the statistics of the outcomes, given the settings". The distribution of settings $x,y$ will never be used: for the purposes of this text, one could just as well treat the settings as parameters labeling probability distribution and write $P_{x,y}(a,b)$ instead.}
\ba\label{genbip}
{\cal P}_{\cal X,\cal Y}&=&\big\{P(a,b|x,y)\,,\,a\in{\cal A},b\in{\cal B}\big\}_{x\in{\cal X},y\in{\cal Y}}\,.
\ea
The statistics ${\cal P}_{\cal X,\cal Y}$ will be called the \textit{observed statistics}: strictly speaking, any experiment involves a finite number of samples and what is really observed is an estimate, whose accuracy can be quantified with the usual statistical techniques.

\subsubsection{Describing the observed statistics}

Without loss of generality, we can write
\ba
P(a,b|x,y)&=&\int d\lambda\,\rho(\lambda|x,y)\,P(a,b|x,y,\lambda)\label{myfavorite}
\ea where $\rho(\lambda|x,y)\geq 0$, $\int d\lambda\,\rho(\lambda|x,y)=1$ and where all the $P(a,b|x,y,\lambda)$ are valid probability distributions. Here, $\lambda$ can be called ``one's favorite explanation": the mathematical description of the process one wants to invoke to explain the observed statistics.

For instance, if \textit{quantum theory} is one's favorite explanation, one will look for a state $\rho$, for $M_A$ POVMs ${\cal M}^x=\{E_a^x|a\in{\cal{A}}\}$ and for $M_B$ POVMs ${\cal M}^y=\{E_b^y|b\in{\cal{B}}\}$, such that $\rho(\lambda|x,y)=\delta(\lambda-\rho)$ and $P(a,b|x,y,\lambda)=\textrm{Tr}(\lambda\,E_a^x\otimes\,E_b^y)$. Thus, one obtains the familiar expression from quantum theory
\ba
P_Q(a,b|x,y)&=& \textrm{Tr}(\rho\,E_a^x\,\otimes\,E_b^y)\,.\label{qprob}
\ea
Another important class of explanations are \textit{deterministic explanations}, in which the outcomes are uniquely determined by the inputs:
\ba
P(a,b|x,y,\lambda)\stackrel{D}{=}\delta_{(a,b)=F(x,y,\lambda)}\,=\,\delta_{a=f(x,y,\lambda)}\delta_{b=g(x,y,\lambda)}\,.\label{expdet}
\ea
The second expression says that, if the pair $(a,b)$ is uniquely determined from the input, then $a$ is uniquely determined and $b$ is uniquely determined.

\subsubsection{Last synopsis before the real start}

In the remainder of this text, unless stated otherwise, I shall assume that \textit{quantum theory gives accurate predictions} or, in other words, that all the \textit{observed} statistics can be obtained with quantum theory and written in the form (\ref{qprob}). The goal is to put quantum theory to the test by not taking it as our favorite \textit{explanation} and trying alternative ones. As we are going to see, all the alternatives that have been tried fail at some point. This brings two scientific benefits: a falsification of the apparently reasonable ideas that lead to formulate the alternative explanations; and a strengthening (if not a confirmation) of quantum theory itself\footnote{While the first benefit is undeniable, the second is debatable: in absolute terms, quantum physics is first and foremost strengthened by the unparalleled scope of its successful predictions. However, there is a benefit in addressing the core tenets of quantum theory in a simple and direct way, instead of just being crushed under the sheer mass of its achievements --- and, in my experience, this is certainly the way to follow in outreach beyond the physics community.}.

Concretely, there is general agreement that the explanation of Bell experiments fully compatible with classical prejudices is given by the so-called \textit{local-variable (LV) models}. Most of this text is devoted to presenting this classical explanation, its failure to reproduce observed statistics (captured by the observation of the violation of Bell inequalities) and the consequences that this fact entails. Once a fully classical explanation is ruled out, one can try to save at least some features of our classical intuition. Remarkably, even those such models that have been proposed fail to reproduce all the predictions of quantum theory (more in Section \ref{secrobust}).

\subsection{Pre-established agreement a.k.a. local variables (LV) a.k.a. shared randomness}

Correlations between distant parties are commonplace in our classical world. For instance, all the agents of a bank delegated to various stock markets start simultaneously selling or buying some bonds, depending on an input which could be the result of a political election. There is no miracle in this simultaneity: either they all received a message from the central office or the senior agent; or, even more probably, they had agreed in advance on how to behave. This example illustrates the only \textit{two classical mechanisms that explain correlations between distant parties: communication (a.k.a. signaling) and pre-established agreement}. Let us leave signaling aside for the moment and describe pre-established agreement.

\subsubsection{Formal characterization of local variables}
\label{deflv}

The idea of pre-established agreement is that the behavior of the parties must be planned in advance for every possible pair of inputs, and in such a way that each party can produce the outcome upon receiving his/her input and not the other party's. Explicitly:
\begin{enumerate}
\item In each run, $\lambda$ may contain information about everything but the values of $x$ and $y$ to be used in that run. Formally, \ba \rho(\lambda|x,y)&=&\rho(\lambda)\,,\label{measind}\ea a condition usually called \textit{measurement independence}.
\item Given $\lambda$, the process that generates $a$ may be stochastic but should not depend on $y$; and the process that generates $b$ may be stochastic but should not depend on $x$. Formally,
\ba
P(a,b|x,y,\lambda)&=&P(a|x,\lambda)P(b|y,\lambda)\,.\label{lhvdef}
\ea
\end{enumerate}

Putting everything together, we obtain
\ba
P_{LV}(a,b|x,y)&=&\int d\lambda\,\rho(\lambda)\,P(a|x,\lambda)P(b|y,\lambda)\,.\label{lvprob}
\ea I have just used the notation ``LV" to mean \textit{local variables}, the expression used in physics to refer to pre-established agreement, which I shall use for its compactness though I find it mildly confusing\footnote{The traditional expression, still widely used, is \textit{local hidden variables}, which carries its unfortunate weight of mysticism. Notably, the adjective ``hidden" is misleading: as we saw, $\lambda$ may describe all kind of information, public or hidden, as long as it does not correlate with $x$ and $y$. The adjective ``local", it is meant to convey two messages: (i) the pre-established agreement could have been worked out when all the agents were at the same location (in physics, when the signals to be measured were created in the source); and (ii) later, each agent acts according only to the information available at its own location. ``Locality" is therefore relevant, but must be understood in a precise sense.}. In computer science, pre-established agreement is usually called \textit{shared randomness}.

Since the notion of LV is central to this course, it is important to appreciate its features, at the risk of being mildly redundant:
\begin{itemize}
\item At first sight, it may seem counter-intuitive that correlations are created by averaging over the $\lambda$'s, because naively one is accustomed to think that statistical distributions wash correlations out. Of course, it depends on \textit{what} one averages over: averaging over the outcomes definitely washes correlations out; averaging over the instructions $\lambda$ is the origin of correlations in scenarios of pre-established agreement.
\item The fact that the statistics of $a$ ($b$) should not depend on $y$ ($x$) is called \textit{no-signaling condition}. Indeed, if Alice could guess some information about $y$ by observing $a$, Bob, who can choose $y$, could send her a message through that channel. When it comes to pre-established agreement,
\ba P(a|x,y,\lambda)=P(a|x,\lambda)&,&P(b|x,y,\lambda)=P(b|y,\lambda)\label{defns}\ea holds for each $\lambda$; but for the observed $P(a,b|x,y)$ to be no-signaling, it is also necessary that \eqref{measind} holds
\item There has been a lot of debate on how best to present the pre-established agreement condition. I tried to keep the presentation to its minimal aspects\footnote{One approach, initiated by Jarrett in 1988, splits \eqref{lhvdef} in two conditions called ``outcome independence" and ``parameter independence", the latter being basically no-signaling (see e.g. \cite{Hall2011}). I adopted this approach in the first version of this course, but I have come to think that it is artificially complicated. Another approach maintains that one must present the two measurement processes in a spacelike separated configuration, where $a$ depends only on what is in its past light cone, and $x$ is chosen as to be only in $a$'s past light cone but not in $b$'s (see e.g. \cite{norsen}). This presentation is indeed striking, but I don't agree that spacelike separation is part of the core assumptions (more in \ref{bellspace}). For further inspiring reading, one can refer to \cite{gisinbellprize,spekkens,gill2012}.}.
\end{itemize}

\subsubsection{Deterministic local variables}
\label{ssdet}

In the example, we have naturally assumed that the agents know exactly what to do when they get to know the winner of the elections. No such restriction was imposed on the mathematics: $P(a|x,\lambda)$ and $P(b|y,\lambda)$ have only been required to be valid probability distributions. \textit{Deterministic local variables} are a special case of LV in which, for any input, the outcome is fully determined by $\lambda$:
\ba
P(a|x,\lambda)\stackrel{DL}{=}\delta_{a=f(x,\lambda)}&,& P(b|y,\lambda)\stackrel{DL}{=}\delta_{b=g(y,\lambda)}\,.
\ea An equivalent way of characterizing deterministic local variables consists in just giving the list of outcomes for all possible inputs:
\ba
\lambda&\stackrel{DL}{\equiv}& \{a_1,a_2,...,a_{M_A};b_1,b_2,...,b_{M_B}\}\,\in{\cal A}^{|{\cal X}|}\times {\cal B}^{|{\cal Y}|}\,,\label{lambdalv}
\ea the link with the previous notation being $a_x=f(x,\lambda)$ and $b_y=g(y,\lambda)$. From this notation, it is obvious that the number of deterministic local points is $m_A^{M_A}\,m_B^{M_B}$. 

The importance of deterministic LV is provided by the following observation, first proved in \cite{fine}:
\begin{proposition}\label{prop21}
A family of probability distributions ${\cal P}_{\cal X,\cal Y}$ can be explained with pre-established agreement if and only if it can be explained with deterministic local variables.
\end{proposition}
\begin{proof} The ``if" direction is obvious. For the reverse, for any fixed $\lambda$, we are going to construct a deterministic model that gives the same statistics as the initial stochastic model\footnote{An interesting mistake (because I committed it myself!) is to believe that the proof is trivial, because it would just be the fact that any probability distribution can be seen as the result of ignorance. This reasoning would lead to the decomposition $P(a|x,\lambda)=\sum_{\mu}q_\mu(x,\lambda)\, \delta_{a=f(x,\mu,\lambda)}$ where $\mu$ labels one of the deterministic points. But $q_\mu(x,\lambda)$ depends on $x$ a priori: if we insert this expression, and the analog one for B, into \eqref{lvprob}, we get a dependence on $x$ and $y$ in the distribution $\rho'(\lambda,\mu_A,\mu_B)$. This is why the proof is not entirely trivial.}. Let's label ${\cal A}=\{1,2,...,m_A\}$ for convenience: the cumulative distribution $\Sigma(a)=\sum_{\alpha\leq a}P(\alpha|x,\lambda)$ can obviously be computed in the LV model. Let us then add a new local parameter $\mu_A$ drawn at random between 0 and 1, then output $a$ according to the following deterministic rule:
\ba
P_d(a|x,\lambda,\mu)&=&\left\{\begin{array}{ll} 1&\textrm{if } \Sigma(a-1)\leq \mu_A < \Sigma(a)\,,\\ 0 &\textrm{otherwise.}
\end{array}\right.\label{pdetproof}
\ea
If $\mu$ is drawn with uniform distribution, the original stochastic distribution is recovered:
\ban
\int_0^1 d\mu\,P_d(a|x,\lambda,\mu)&=& \int_{\Sigma(a-1)}^{\Sigma(a)} d\mu \,=\,P(a|x,\lambda)\,.
\ean
Therefore \eqref{lvprob} can be rewritten as
\ba
P_{LV}(a,b|x,y)&=&\int d\lambda \rho(\lambda)\int_0^1 d\mu_A \int_0^1 d\mu_B\,P_d(a|x,\lambda,\mu_A)P_d(b|y,\lambda,\mu_B)\,,\label{lvdet}
\ea
which is the desired \textit{convex sum} of deterministic LV for the enlarged variable $\lambda'\equiv(\lambda,\mu_A,\mu_B)$ with distribution $\rho'(\lambda')d\lambda'=\rho(\lambda)d\lambda d\mu_A d\mu_B$. \end{proof}

It follows immediately that the finite set of $\lambda$'s defined by the $m_A^{M_A}\,m_B^{M_B}$ deterministic local point is sufficient to describe any LV statistics. Indeed, each $P_d(a|x,\lambda,\mu)$ given in \eqref{pdetproof} is one of the $m_A^{M_A}$ deterministic points for A; and similarly for $B$. By grouping the terms according to the local deterministic points, \eqref{lvdet} becomes
\ba
P_{LV}(a,b|x,y)&=&\sum_{j=1}^{m_A^{M_A}} \sum_{k=1}^{m_B^{M_B}} \rho_{jk}\,\delta_{a=f_j(x)}\,\delta_{b=g_k(y)}
\ea with $\lambda\equiv(j,k)$ and $\sum_{j,k} \rho_{jk}=1$.

Therefore LV statistics can always be explained by a deterministic model. Of course, this does not mean that such an explanation must necessarily be adopted: your favorite explanation, as well as the ``real" phenomenon, may not involve determinism. For instance, as we shall see soon, measurement on separable quantum states leads to LV statistics, but this does not make quantum theory deterministic (if that is your favorite explanation), nor forces us to believe that the physical phenomenon ``out there" is deterministic.

Finally, Proposition \ref{prop21} has an interesting translation using the notation \eqref{lambdalv}:
\begin{corollary}
${\cal P}_{\cal X,\cal Y}$ can be explained with pre-established agreement if and only if there exist a \textit{joint probability distribution} $\mathbf{P}(a_1,...a_{M_A};b_1,...,b_{M_B})$, such that each $P(a,b|x,y)$ can be obtained as the two-party marginal
\ba
P_{LV}(a,b|x,y)&=&\sum_{\{a_j|j\neq x\}}\sum_{\{b_k|k\neq y\}}\mathbf{P}(a_1,...a_{M_A};b_1,...,b_{M_B})\,.\label{lvmarg}
\ea
\end{corollary}

\subsection{The power of LV}
\label{powerlv}

At this stage, physicists usually hurry forward and prove that quantum statistics cannot be explained with LV. This is of course my goal as well, but too much haste in taking this step may convey the impression that the LV models are silly and artificial, and that it is only normal that quantum correlations don't fit in that class. I want to devote one paragraph to explain what LV models could do for you, beyond explaining coordinated behavior in stock market bidding.

\begin{itemize}
\item We have described the Bell experiment using two parties because \textit{quantum statistics involving one party can always be reproduced with LV}; which means, according to what we just proved, that they can even be reproduced by a \textit{deterministic} explanation. The reason is that $\lambda$ may contain the description of the quantum state $\rho$: one can compute on paper the probability distribution for any measurement, then just draw the outcomes according to that distribution using pseudo-randomness. As the opening parable illustrates, this statement is obvious once one thinks freely about it, but is nevertheless puzzling for someone (the student) who has been (de)formed in thinking that single-party phenomena already unveil the intrinsic randomness of the quantum. To be sure, one can discover many quantum effects in single-party measurements, but not in a black-box scenario: further assumptions are needed\footnote{The \textit{knowledge of the physical degree of freedom} usually provides the additional assumption. For instance, the blackbody radiation or the Stern-Gerlach experiment can be proved to be non-classical once one knows that the electromagnetic field, respectively a magnetic moment, is being measured. Similarly, there are often very good reasons to \textit{describe the system as composite}, even if all the components contribute to the same signal: the prototypical examples come from condensed matter physics, where one wants to describe conductivity, magentization etc as arising not from an unspecified piece of matter, but from an arrangement of many atoms. Tests like ``contextuality" \`a la Kochen-Specker, or sequential measurements \`a la Feynman or Leggett-Garg, need a minimal amount of assumptions to prove that the outcomes do not come from a pre-established agreement. Indeed, no detailed knowledge of the system and the measurement is needed, but one must assume that \textit{the measurement device does not write in, nor reads from, other degrees of freedom than the relevant one}.}.

Though anecdotical, it is instructive to present a pretty compact LV model that reproduces exactly the quantum statistics of a \textit{single qubit}. If the qubit is in the state $\frac{1}{2}(1+\vec{m}\cdot\vec{\sigma})$, the quantum prediction for measurement along direction $\vec{a}$ is
\ba
P(a|\vec{a})=\frac{1}{2}(1+a\vec{m}\cdot\vec{a})\,,&\textrm{that is}& \moy{a}_{\vec{a}}=\vec{m}\cdot\vec{a}
\ea for $a\in\{-1,+1\}$. In the LV model, the system is represented by $\vec{m}$ and by a unit vector $\vec{\lambda}$ uniformly distributed\footnote{Nothing in this model requires $\vec{\lambda}$ to be ``drawn at random" in each run: the sequence of $\vec{\lambda}$ may be pre-registered.} on the sphere, i.e. $\rho(\vec{\lambda})d\vec{\lambda}=\frac{1}{4\pi}\sin\theta d\theta d\varphi$. For each $\vec{\lambda}$, the outcome is deterministically computed to be
\ba
a(\vec{\lambda})&=&\textrm{sign}\big[(\vec{m}-\vec{\lambda})\cdot\vec{a}\big]
\ea which is either $+1$ or $-1$ as it should. It is then easy to prove\footnote{Without loss of generality, one can choose $\vec{a}=\vec{z}\equiv (\theta=0,\varphi)$, since nothing else in the problem specifies the choice of spherical coordinates.} that
\ba
\moy{a}_{\vec{a}}&=&\int d\vec{\lambda}\,\rho(\vec{\lambda})\, a(\vec{\lambda})\,=\,\vec{m}\cdot\vec{a}\,.
\ea
Another nice illustrative example is the calculation of the \textit{double slit experiment with Bohmian trajectories}, which reproduces the quantum interferences while being also able to say through which slit the particle has passed. To conclude on a personal note: if quantum physics would consist only of single-party measurements, I would not see any compelling reason to believe in its intrinsic randomness.

\item LV can also simulate several bipartite scenarios. One of the most obvious ${\cal P}_{\cal X,\cal Y}$ that can be described with LV is exactly the one that is presented in popular lore as an astonishing feat of quantum entanglement: the fact that, when they share a singlet state, \textit{Alice and Bob observe always opposite results when they measure in the same direction}, i.e.
\ba
\moy{ab}_{\vec{a}=\vec{b}}=-1\,.\label{anticorr}
\ea Indeed, any list
\ba
\lambda&=&\big\{..., a_{\vec{u}},a_{\vec{v}},...; ..., b_{\vec{u}},b_{\vec{v}},...\big\}
\ea deterministically obtains (\ref{anticorr}) as soon as $a_{\vec{u}}=-b_{\vec{u}}$ for all possible directions $\vec{u}$.

If now we want to create correlations, we need to average over various $\lambda$. For the sake of illustration, let us restrict ourselves to two possible measurements on each side, so $\lambda=\{a_{\vec{u}},a_{\vec{v}};b_{\vec{u}},b_{\vec{v}}\}$. In order to satisfy (\ref{anticorr}) we can have
\ba
\begin{array}{lcl}
\lambda_1&=&(+,+;-,-)\\
\lambda_2&=&(+,-;-,+)\\
\lambda_3&=&(-,+;+,-)\\
\lambda_4&=&(-,-;+,+)
\end{array}
\ea each drawn with probability $\mathbf{P}(\lambda_k)\equiv q_k$. Moreover, any choice such that $q_1=q_4$ and $q_2=q_3=\demi-q_1$ reproduces $\moy{a}_{\vec{u}}=\moy{a}_{\vec{v}}=0$. By choosing $q_1=\frac{1}{4}(1+\vec{u}\cdot\vec{v})$, we can even reproduce the quantum predictions for the case where the two measurements are not the same, i.e. $\moy{ab}_{\vec{u},\vec{v}}=\moy{ab}_{\vec{v},\vec{u}}=-\vec{u}\cdot\vec{v}$.

This result is more striking than its simplicity suggests. Consider the case $\vec{u}\perp\vec{v}$, for instance $\vec{u}=\vec{x}$ and $\vec{v}=\vec{y}$: the statistics we have just reproduced with LV are the expected correlations for an error-free run of the BB84 protocol for quantum key distribution\footnote{In BB84, one expects \textit{correlations} rather than anticorrelations; but it is trivial to obtain the former from the latter by classical post-processing: for instance, one can ask Bob to flip systematically his bit.}! This means that the security proofs of that protocol must be based on addditional assumptions, other than the mere observation of those statistics, and indeed upon closer inspection it assumes that qubits are being measured (this observation was crucial in the emergence of the device-independent outlook, see Appendix \ref{sspathdi}).

Ultimately, everything will fall into place: we are going to prove that the whole set of quantum predictions $\moy{ab}_{\vec{a},\vec{b}}=-\vec{a}\cdot\vec{b}$ for \textit{any} pair of directions \textit{cannot} be reproduced with LV. Before that, please bear with one more example of the power of LV --- and above all, never again invoke perfect anticorrelations as an evidence for entanglement.

\item Our second example of bipartite statistics that can be simulated with LV are those in which \textit{one of the parties performs only one measurement}. Once again, when one thinks about it, it is clear that this should be the case. Indeed, the set of \textit{possible} measurements for both boxes can be part of the information in $\lambda$. If Alice's set contains a single element, $|{\cal X}|=1$, then Alice's \textit{actual} measurement in each individual run is also known. In this case, $\lambda$ can determine Alice's outcome and distribute Bob's outcome accordingly for all his possible measurements. Therefore, in order to rule out an explanation by pre-established agreement, one must consider families ${\cal P}_{\cal X,\cal Y}$ such that both $|{\cal X}|>1$ and $|{\cal Y}|>1$. 

\end{itemize}

\subsection{Bell inequalities and their violation}

\subsubsection{Bell inequalities as facets of the local polytope}

The starting point for our study of Bell inequalities is the following observation: for any fixed scenario $({\cal X},{\cal A};{\cal Y},{\cal B})$, \textit{the set ${\cal L}$ of all the families of probability distributions that can be obtained with LV is convex}. In other words, if ${\cal P}_{1}\in{\cal L}$ and ${\cal P}_{2}\in{\cal L}$, then $q{\cal P}_{1}+(1-q){\cal P}_{2}\in{\cal L}$ for all $q\in [0,1]$. This is clear from the interpretation: $\lambda$ can contain the information of whether ${\cal P}_{1}$ or ${\cal P}_{2}$ is realized in each run. Presently we need to describe the convex set ${\cal L}$ in more detail.

A convex set is fully determined by specifying all its extremal points, i.e. those points that cannot be written as convex combinations of other points. We know from Proposition \ref{prop21} that any ${\cal P}\in{\cal L}$ can be written as a convex sum of deterministic LV, and it is easy to convince oneself that \textit{each deterministic local point is an extremal point of ${\cal L}$}. Moreover, there are finitely many such points, precisely $m_A^{M_A}\,m_B^{M_B}$, as we noted above. A convex set with finitely many extremal points is concisely referred to as ``polytope", so ${\cal L}$ will be called the \textit{local polytope} for the scenario $({\cal X},{\cal A};{\cal Y},{\cal B})$.

A polytope ${\cal L}$ embedded in $\real^D$ is delimited by $(D-1)$-dimensional hyperplanes called \textit{facets}. Such a hyperplane must have at least $D$ extremal points lying on it, while all the other extremal points must be found on the same side (if some extremal points are on one side and others on the other, then the hyperplane cuts through the polytope and is not a facet). Mathematically, let $n\in\real^D$ is the vector normal to the facet and oriented outside the polytope: if the equation for the points ${\cal P}$ of the facet is $n\cdot{\cal P}=f$, then
\ba
n\cdot{\cal P}\,\leq\,f&\textrm{for all}& {\cal P}\in{\cal L}\,.\label{bigeom}
\ea
In other words, if $n\cdot{\cal P}>f$, the point ${\cal P}$ cannot belong to ${\cal L}$.

In the specific case of a probability polytope like ${\cal L}$, some facets are given by the equations $P(a,b|x,y)=0$ and $P(a,b|x,y)=1$. Such facets are called \textit{trivial} because they are not proper to ${\cal L}$, indeed they describe the constraints $0\leq P(a,b|x,y)\leq 1$ that any probability distribution must satisfy. There must be other facets, however, which capture the constraints proper to ${\cal L}$. \textit{The inequalities (\ref{bigeom}) associated to the non-trivial facets of ${\cal L}$ are the Bell inequalities\footnote{To set history straight, such inequalities were noticed pretty early in statistics: Boole certainly describes them. But in those classical days, neither he nor anyone else considered the possibility of their violation. By contrast, John Bell derived a single inequality, which is not even a facet but a lower dimensional hyperline on a facet, because he used additional constraints in the derivation. Nevertheless, he made the great step of pioneering the method in the study of quantum physics; so, at least in physics, inequalities of this type are generically named after him.} for the scenario under study}.

Before studying a specific example, we have to determine the minimal $D$ such that ${\cal L}\subset \real^D$. Generically, a family ${\cal P}_{\cal X,\cal Y}$ contains $M_AM_B$ probability distributions, each of which is specified by $m_Am_B$ positive numbers constrained to sum to 1; so it can be fully specified by giving $D_{\mathrm{total}}=M_AM_B(m_Am_B-1)$ independent numbers. The ${\cal P}_{\cal X,\cal Y}$ that can be reproduced by LV satisfy the additional \textit{no-signaling} constraints, namely $\sum_bP(a,b|x,y)=\sum_bP(a,b|x,y')=P(a|x)$ for all $y,y'\in{\cal Y}$, and $\sum_aP(a,b|x,y)=\sum_aP(a,b|x',y)=P(b|y)$ for all $x,x'\in{\cal X}$. The counting can be done as follows. One can first take the marginals as independent parameters: there are $M_A(m_A-1)$ independent $P(a|x)$, and $M_B(m_B-1)$ independent $P(b|y)$. Consider now any choice of $(x,y)$: for every fixed $b=\bar{b}$, once the marginal $P(a)$ is given, one is left with $m_A-1$ independent numbers $P(a,\bar{b})$; similarly, for every fixed $a=\bar{a}$, once the marginal $P(b)$ is given, one is left with $m_B-1$ independent numbers $P(\bar{a},b)$. All in all, a no-signaling ${\cal P}_{\cal X,\cal Y}$ can be fully specified by giving
\ba
D_{\mathrm{NS}}&=&M_AM_B(m_A-1)(m_B-1)+M_A(m_A-1)+M_B(m_B-1)\label{dns}
\ea independent numbers. Since our goal is to compare LV with quantum physics, which is also no-signaling, it is sufficient to describe \textit{$\cal L$ as embedded in $\mathbb{R}^{D_{\mathrm{NS}}}$}.

\subsubsection{A case study: CHSH}
\label{ssscase}

It is instructive to work out explicitly the facets of $\cal L$ and derive the corresponding Bell inequalities. The simplest scenario has $M_A=M_B=m_A=m_B=2$. In this case, $D_{\mathrm{NS}}=8$ and there are 16 extremal points: finding the facets is a very easy task for a computer, but still cumbersome to write down here. We are rather going to study a very meaningful sub-polytope of $\cal L$.

For a choice of settings $(x,y)$ and binary outcomes, the \textit{correlation coefficient} is defined by
\ba
E_{xy}&=&P(a=b|x,y)-P(a\neq b|x,y)\,.\label{corrcoeff}
\ea Any quadruple of numbers
\ba
u&=&(E_{00},E_{01},E_{10},E_{11})
\ea with $-1\leq E_{xy}\leq 1$, is \textit{a priori} a valid correlation vector. The sixteen vectors such that $||u||^2=4$, i.e. those vectors whose components are either $+1$ or $-1$, are extremal points of a polytope embedded in $\mathbb{R}^{4}$.

To see which constraints are added by requiring that ${\cal P}\in{\cal L}$, it is convenient to use the labeling convention $a,b\in\{-1,+1\}$. With this choice, $E_{xy}=\moy{a_xb_y}$. In particular, for deterministic local points it holds $E_{xy}\stackrel{D}{=}a_xb_y$, which directly leads to $E_{00}E_{01}E_{10}E_{11}\stackrel{D}{=}1$. Therefore, the extremal points of the \textit{local correlation polytope} are the eight vectors
\ba
\begin{array}{lcl}
v_1=(+1,+1,+1,+1)&,& v_5=-v_1=(-1,-1,-1,-1)\\
v_2=(+1,+1,-1,-1)&,& v_6=-v_2=(-1,-1,+1,+1)\\
v_3=(+1,-1,+1,-1)&,& v_7=-v_3=(-1,+1,-1,+1)\\
v_4=(+1,-1,-1,+1)&,& v_8=-v_4=(-1,+1,+1,-1)\,.
\end{array}
\ea Notice that $\{v_1,v_2,v_3,v_4\}$ are mutually orthogonal, so in particular they are linearly independent: this implies that $\mathbb{R}^{4}$ is the smallest embedding for the local correlation polytope. Now we want to characterize the facets of this polytope. Four linearly independent vectors are required to define a 3-dimensional hyperplane\footnote{This is because the origin $(0,0,0,0)$ is inside the polytope. If the origin would be on a facet, only three linearly independent vectors would be sufficient to specify such a facet. This is a technical point that does not need to bother us here, but may play a role for most compact parametrizations of probability polytopes.}: our task consist of listing all sets of four linearly independent extremal points, constructing the hyperplane that they generate, and checking if it is indeed a facet.

The symmetry of the problem makes the task simple. There are sixteen sets of four linearly independent vectors, namely the $V_{\underline{s}}=\{s_1v_1,s_2v_2,s_3v_3,s_4v_4\}$ with $\underline{s}=[s_1,s_2,s_3,s_4]\in\{-1,+1\}^4$. The normal to the hyperplane generated by $V_{\underline{s}}$ is the solution to the equation $n_{\underline{s}}\cdot(s_kv_k)= f_{\underline{s}}\equiv 4$ (the constant being chosen for simplicity), which is readily found to be $n_{\underline{s}}=\sum_{k=1}^4s_kv_k$, either by direct inspection or by noticing that $v_i\cdot v_j=4\delta_{ij}$. Moreover, the extremal points that do not define the plane are the four $-s_kv_k$, which all lie on the same side of the hyperplane since obviously $n_{\underline{s}}\cdot(-s_kv_k)=-4$. Therefore each of the sixteen sets $V_{\underline{s}}$ defines a facet by the condition
\ba
n_{\underline{s}}\cdot u\,=\, 4&\textrm{with}& n_{\underline{s}}=\sum_{k=1}^4s_kv_k\,.
\ea
To finish our study, we just have to inspect each of these facets (since $n_{-\underline{s}}=-n_{\underline{s}}$, we can just look at eight of them). We find:
\ba
\begin{array}{lcl}
n_{[+1,+1,+1,+1]}=(4,0,0,0)&\longrightarrow & 4E_{00}\,\leq\,4\,,\\
n_{[+1,+1,+1,-1]}=(2,2,2,-2)&\longrightarrow & 2E_{00}+2E_{01}+2E_{10}-2E_{11}\,\leq\,4\,,\\
n_{[+1,+1,-1,+1]}=(2,2,-2,2)&\longrightarrow & 2E_{00}+2E_{01}-2E_{10}+2E_{11}\,\leq\,4\,,\\
n_{[+1,-1,+1,+1]}=(2,-2,2,2)&\longrightarrow & 2E_{00}-2E_{01}+2E_{10}+2E_{11}\,\leq\,4\,,\\
n_{[+1,+1,-1,-1]}=(0,4,0,0)&\longrightarrow & 4E_{01}\,\leq\,4\,,\\
n_{[+1,-1,+1,-1]}=(0,0,4,0)&\longrightarrow & 4E_{10}\,\leq\,4\,,\\
n_{[+1,-1,-1,+1]}=(0,0,0,4)&\longrightarrow & 4E_{11}\,\leq\,4\,,\\
n_{[+1,-1,+1,+1]}=(-2,2,2,2)&\longrightarrow & -2E_{00}+2E_{01}+2E_{10}+2E_{11}\,\leq\,4\,.
\end{array}
\ea
In summary, the local correlation polytope for the simplest scenario has sixteen facets. Up to relabeling of the inputs and/or of the outcomes, eight of these describe the constraint $E_{00}\leq 1$ and are therefore trivial, while the other eight describe the constraint
\ba
S\equiv E_{00}+E_{01}+E_{10}-E_{11}\,\leq\,2\,.\label{chshcorr}
\ea
This constraint is not trivial, and indeed can be violated by valid correlation vectors which do not belong to the local polytope: in particular, the vector $w=(+1,+1,+1,-1)$ reaches up to $S=4$. The inequality (\ref{chshcorr}) is called \textit{CHSH} from the names of Clauser, Horne, Shimony and Holt, who derived it first in physics \cite{chsh}. It is the most studied of all Bell inequalities\footnote{Bell's original inequality \cite{bell64} is ultimately CHSH, but he did not derive it using the systematic approach we just sketched. Rather, he had in mind measurements on a two-qubit singlet state, so he imposed that the LV model should satisfy \eqref{anticorr} exactly. This lead to an expression that is valid only under that assumption; which is enough in order to prove that quantum \textit{theory} violates the inequality, but is not suitable for comparison with experiments, since one will never observe absolutely perfect (anti)correlations.} and we shall use it repeatedly in this text.

Had we studied the full local polytope ${\cal L}$, we would have found only eight more facets, describing the trivial constraints $0\leq P(a=0|x)\leq 1$ and $0\leq P(b=0|y)\leq 1$ on the marginals, which certainly cannot be captured from the correlations alone. In conclusion, \textit{CHSH is the only Bell inequality in the scenario $M_A=M_B=m_A=m_B=2$}.

\subsubsection{Detection loophole: faking it by denial of service}

We have adopted an operational approach, in which the violation of Bell inequalities is read as the impossibility of a fair simulation of the observed statistics. From this perspective, something that is normally (and rightly) considered as anecdotical in more physically-based presentations acquires a huge importance: the possibility of faking a violation of Bell inequalities by a clever denial of service. We can discuss it now with the example of CHSH.

We have just seen that the correlation vector $w=(+1,+1,+1,-1)$ reaches the maximal possible value $S=4$. The LV vectors $v_1$, $v_2$, $v_3$ and $-v_4$ differ from $w$ only for one choice of settings: for instance, $v_1$ behaves like $w$ as long as the pair $(x,y)=(1,1)$ is not chosen. Because of measurement independence, $\lambda$ cannot guarantee that $(1,1)$ won't be chosen in that run. However, $\lambda$ can prevent such an event to be seen by adding the following instruction: for $x=1$, Alice's box does not reply. Now, if Alice and Bob naively estimate CHSH only from those instances in which both have got a reply, they may easily be cheated into obtaining $S>2$ (and even $S=4$) whereas only LV were used. If the bureaucrat and his friend had been aware of this possibility and the student had not, the two men might have faked a successful simulation. This possibility is called \textit{detection loophole} in the physics jargon.

This proves that \textit{post-selection is not an allowed data processing when dealing with Bell inequalities} on a purely operational basis\footnote{In the normal working of physics, which is the study of nature with characterized devices and without conspiracy theories, post-selection is licit as soon as one knows how the detector works and knows that the causes of reduced efficiency are completely independent from the choice of the measurement setting. In this sense, we can very safely claim that the violation of Bell inequalities \textit{has been observed} (in fact, far more often than many other physical effects, that are the object of less close scrutiny because their consequences are far less deep). However, if it comes to characterizing untrusted black boxes, or to convince a skeptical, then post-selection is not allowed.}. The remedy to avoid the trap is clear: Alice and Bob must compute the statistics from the whole sample. When the boxes (in physics, the detectors) do not give any answer, Alice and Bob can adopt \textit{a priori} either of two strategies. The first consists in treating the lack of answer as another outcome, in which case the scenario changes from $(m_A,m_B)$ to $(m_A+1,m_B+1)$. The second consists in agreeing in advance that the lack of answer will be treated as one of the outcomes (say $+1$, always the same). Which processing is more efficient may depend on the scenario; but in both cases, \textit{the statistics are evaluated on all the runs}: a violation of Bell inequalities by those statistics is conclusive.

\subsubsection{Forgetting memory}
\label{ssgill}

At the very beginning of this study, I stated the temporary assumption that the black boxes behave in the same way in each run. There is no compelling reason for this to be the case: at run $i$, the boxes may well be taking into account the inputs and outputs of the previous $i-1$ runs. However, now we have collected enough knowledge to understand that such \textit{memory effects cannot be used to fake a violation of Bell inequalities}.

Indeed, let us describe the most general protocol based on LV. We know already that the boxes can be assumed to act according to deterministic instructions $\lambda_i$ in each run $i$. So far, we had assumed that $\lambda_i$ is drawn independently in each run. Let us now drop this assumption and allow $\lambda_i$ to depend on all that happened in the previous $i-1$ runs. Even now, however, \textit{in each run, each box is prepared with deterministic instructions}. We could keep the reasoning general, but I find it more instructive here to refer explicitly to CHSH. Recall that $E_{xy}=a_xb_y$ if we choose $a,b\in\{-1,+1\}$. The maximum of $S$ is reached if $a_0b_0=+1$, $a_0b_1=+1$, $a_1b_0=+1$ and $a_1b_1=-1$. A local deterministic $\lambda_i$ can satisfy at most three of these conditions, at the price of getting the fourth one completely wrong; whence the LV bound $S_L=3-1=2$.

Now, the crucial observation is that, in a Bell test, Alice and Bob choose $x$ and $y$ independently in each run: in particular, the settings are completely uncorrelated from the $\lambda$'s, including the past ones. This means that, irrespective of how $\lambda_i$ is chosen, Alice and Bob may choose the ``wrong" pair of settings (and the total probability with which this happens does not matter, since $S$ is estimated from \textit{conditional} probabilities). Therefore, memory effects cannot be used to fake a violation of Bell inequalities with LV, as claimed\footnote{In this text, as mentioned, I concentrate on the asymptotic case of statistics gathered in infinitely many runs. In practice, it is essential to deal with finite-size statistics. It's this estimate that must be done more carefully when memory effects are taken into account: see section 2 of \cite{gill2012} for details and references.}.

\subsubsection{The message in the violation}
\label{ssmessage}

We have accumulated all the notions required for studying Bell inequalities in the context of quantum physics --- and indeed, it is worth while stressing very explicitly that \textit{no element of quantum theory has been used so far}. This means in particular that the meaning of the violation of Bell inequalities is theory-independent. In today's scientific jargon, we say that the violations observed in laboratories are due to ``quantum entanglement". Maybe future scientists will use different concepts, but this won't change the fact.

As for what \textit{the meaning of the violation actually is}, rivers of ink have been spent in arguing on that, and even frequently-used terms like ``nonlocality" and ``violation of local realism" are the object of sometimes heated debates\footnote{These terms can indeed be equivocal. I have met students who had memorized the slogan ``the violation of Bell inequalities demonstrates nonlocality" and had got a wrong understanding of it all. Similarly, hearing physicists claiming that ``realism is denied", many conclude that we can't know anything about reality --- while everything started by accepting the reality of the observed statistics. I don't want to impose a ban on any of these terms and feel in fact free to use them myself whenever convenient, but they must be correctly understood.}. I shall keep an economic approach: \textit{if the inequalities are violated, pre-established agreement is not the explanation}. So, either the results were determined by another mechanism, or they were not determined but new information is created. Let us see what one can say about these two options: 
\begin{itemize}
\item In order to keep determinism, $b$ must depend on $x$ or $a$ on $y$. This can happen in two ways:
\begin{itemize}
\item[-] A \textit{signal} sends the input of one party to the other location. This classically plausible explanation is very problematic in the present case. Both quantum theory and experiments agree that the violation of Bell inequalities does not change with distance. By putting enough distance between the box of Alice and that of Bob, so as to ensure that one's choice of input is spacelike separated from the other's output, the hypothetic signal should travel faster than light. Some authors have toyed with the idea of ``peaceful coexistence", arguing that one could think of a signal propagating faster than light, as long as it cannot be used by us to send an actual message. However, if one wants to reproduce all quantum predictions, this requirement can be met only by signals propagating at infinite speed (see Section \ref{secrobust}). These conclusions about the nature of the signal are \textit{phenomenological but (better: and therefore) inescapable}: even if the quantum world is seen as emergent, whatever deeper reality one may like to consider\footnote{Flights of imagination worthy of science fiction may lead one to think that each pair of entangled systems brings with it its own wormhole, its own extra dimension, in which the two objects remain forever close. The infinitely rigid ``quantum potential" of Bohmian mechanics, or t'Hooft's idea of the quantum fields emerging from a big cellular automaton, have a more serious ring to them.} must manifest itself as the phenomenon of infinite-speed signal in our (3+1)-dimensional space-time.

\item[-] The other option is the denial of measurement independence: the preparation of the system in the boxes would then differ according to which measurement are going to be performed on it. I would argue that this choice is incompatible with the practice of science\footnote{The denial of measurement independence is frequently called \textit{denial of free will}. Indeed, if $(x,y)$ are correlated with $\lambda$, then for a given $\lambda$ some pairs $(x,y)$ must be more probable than others. Now, the settings could in principle be chosen by human beings, whose ``free will" would then be maimed. The argument is correct. However, when used, it tends to stir philosophical passions, typically driving the debate towards \textit{The Matrix}, Libet's experiments, or moral responsibility --- all very interesting topics, but whose connection with quantum randomness is (to adopt an optimistic stance) unclear.}. Indeed, a fundamental tenet of the scientific method is the possibility of performing different measurements on \textit{identically prepared} systems. I don't know who would dare denying this possibility in order to save determinism in Bell experiments.
\end{itemize}

\item If all the mechanisms to ensure determinism are problematic, the remaining explanation, however improbable, must be accepted: \textit{the violation of Bell inequalities proves the existence of intrinsic randomness}. Unless stated otherwise, the remainder of this text will be written from this perspective. This is a very strong statement, of course. It is supported by the mathematics of quantum theory, since in general (\ref{qprob}) cannot be written as (\ref{lhvdef}); it fits also the so-called orthodox interpretation\footnote{The supporters of the many-worlds interpretation may want to say that the violation of Bell inequalities proves the existence of intrinsic randomness \textit{in each universe}. In this operationally-oriented text, I don't want to get lost in such subtleties and assume that each of us will experience only one universe throughout their lives, however many other universes may ``really" exist out there.}. Two further points are worth stressing:
\begin{itemize}
\item[-] The violation of Bell inequalities is \textit{the phenomenon} that proves the existence of intrinsic randomness. In other words, at the risk of repeating what I wrote in paragraph \ref{powerlv}: it's only because two degrees of freedom, even separated in space, violate Bell inequality that we can safely infer the presence of intrinsic randomness also when a single degree of freedom is involved, as in the double-slit experiment or in Heisenberg's uncertainty relations.
\item[-] Intrinsic randomness \textit{per se} cannot be \textit{a sufficient explanation} for the violation of Bell inequalities: one can easily conceive a world with intrinsic randomness, in which correlations are nevertheless compatible with LV\footnote{After all, most physicists believed ardently in intrinsic randomness before Bell ended up certifying their belief. I thank Nicolas Gisin for bringing this point to my attention.}. From this perspective, a minimal sufficient explanation consists in postulating that intrinsic randomness must be certifiable; but I prefer to try and recover quantum physics from physical principles such as those described in Section \ref{secprinciples}.
\end{itemize}
\end{itemize}

The last words of this paragraph are dedicated to those who don't want to make the step either to intrinsic randomness or to infinite-speed signals, and hope to recover a situation where all is ``normal". Usually, these skeptics scrutinize the mathematics, looking for the hidden assumption, the one that nobody would have noticed. Such attempts should be replaced with an operational challenge: skeptics should be asked to exhibit a violation of Bell inequalities between two classical computers, without communication and without post-selection\footnote{In order to convince the last die-hard, the rules of the simulation must be stated with great accuracy: the painful task of defining how such a ``Randi test" should look like has been undertaken by others \cite{gill2012,vongehr}.}. The effort of building such a simulation may prove very instructive --- for them.

\subsubsection{Bell and spacelike separation}
\label{bellspace}

I want to conclude by hopefully dissipating some frequent confusions. Let me first repeat that the violation of Bell inequalities \textit{per se} proves that pre-established agreement is not the explanation, and nothing else. In the context of quantum physics (Bohmians aside), one would further seek to rule out the other classical explanation, signaling, thus positively proving the existence of intrinsic randomness. Spacelike separation is the ultimate way of guaranteeing no-signaling, which should convince everyone (again, Bohmians aside); but it is not a strict requirement for a Bell experiment. Let me go back to the parable: how will the student make sure that the bureaucrat and his friend won't communicate during the simulation? She will certainly put them in non-adjacent rooms, out of obvious visual and acoustic contact. Beyond that, she may just trust their integrity, perhaps performing checks at random times; or she may ask them to leave their handphone in her lab; or she may put them in Faraday cages, thus assuming that they can only use electromagnetic waves as signals... Only in an extremely adversarial scenario the student will enforce spacelike separation in the protocol. In summary: in Bell experiments, \textit{there is an operational asymmetry between ruling out pre-established agreement (which is what Bell inequalities can do) and ruling out communication (which is also of great importance but requires independent criteria)}.

\section{Bell inequalities and quantum physics: the very basics}
\label{secquantum}

The content of this section can be summarized by one sentence: \textit{quantum theory predicts the violation of Bell's inequalities and experiments have confirmed this prediction}.

Any closer look will reveal a rather rigged landscape. In the past twenty years or so, a lot of Bell inequalities have been listed, some as belonging to nicely defined families, others as items of otherwise unstructured catalogues. Some of these inequalities lead to genuine refinements over CHSH, which the experts appreciate. Overall though, even these charted regions are largely unexplored and unexploited; and of course we can't know in advance if further exploration will lead to new insights, to a synthetic comprehension, or just to additions to the catalogue\footnote{For instance, some authors have recently discovered a family of Bell inequalities that no quantum state can violate, a result that at first sight looks like as pointless a mathematical exercise as it gets; instead, these inequalities have become a tool to gain quite interesting insights on the structure of quantum statistics. A rule of thumb: as long as you find articles published in journals with generic scope, the field may still be thriving. If you see the birth of a ``Journal of Bell inequalities", you know that the end is approaching.}. Anyway, the basics can be illustrated in the elementary scenario with two parties, each with two inputs and two outputs, leading to the \textit{CHSH inequality}. I shall confine myself to it here and throughout all this text, unless stated otherwise.

\subsection{CHSH operator and Tsirelson bound}

As usual in physics, one can choose the most convenient labeling: notably, if we choose to label the outcomes $a,b\in\{-1,+1\}$, the correlation coefficient defined in \eqref{corrcoeff} becomes simply $E_{xy}=\moy{a_xb_y}$. Here, we want to consider the case where the outcomes $a_x$ and $b_y$ are results of quantum measurements: there exist four Hermitians operators $\hat{A}_0,\hat{A}_1,\hat{B}_0,\hat{B}_1$ with eigenvalues $-1$ and $+1$, possibly degenerate, such that $E_{xy}=\moy{A_x\otimes B_y}$. 

In a Bell-CHSH experiment, each single-shot measurement corresponds to one of the $\hat{A}_x\otimes \hat{B}_y$; the statistics of four such series of measurements are later combined in the form \eqref{chshcorr}. Within quantum theory, because of linearity, the same statistics can be seen as the average value of the \textit{CHSH operator}
\ba
\hat{S}&=&\hat{A}_0\otimes \hat{B}_0+\hat{A}_0\otimes \hat{B}_1+\hat{A}_1\otimes \hat{B}_0-\hat{A}_1\otimes \hat{B}_1\,.
\ea
Even though a Bell experiment does \textit{not} consist of single-shot measurements of $\hat{S}$, the translation of Bell inequalities in the language of \textit{Bell operators} is extremely useful\footnote{Here we started by some remarks on the labeling, but of course one is not obliged to find a particularly clever labeling before writing down Bell operators: any Bell inequality is a linear combination of probabilities, therefore the corresponding operator is obtained by replacing each probablity with the corresponding projectors.}. In particular, \textit{a quantum state $\rho$ violates CHSH if and only if there exist measurements such that $\Tr(\rho\,\hat{S})>2$}.

Before introducing states that actually violate CHSH, let us show a generic result, known as \textit{Tsirelson bound} \cite{tsi}:
\begin{theorem}\label{thm1}
Measurements on quantum systems can violate the CHSH inequality \eqref{chshcorr} at most up to
\ba
S&\leq&2\sqrt{2}\,.\label{tsirelon}
\ea
\end{theorem}

\begin{proof} 
In order to prove the claim, we need to find an upper bound for the largest eigenvalue of $\hat{S}$, denoted as usual by $||\hat{S}||_\infty$. By construction, we have $||\hat{A}_x||_{\infty}=||\hat{B}_y||_{\infty}=1$, $\hat{A}_x^2=\one_{d_A}$ and $\hat{B}_y^2=\one_{d_B}$; the dimensions $d_A$ and $d_B$ of the Hilbert spaces are left unspecified and may be infinite. The bound is very simple to obtain by working with the square of the CHSH operator
\ba
\hat{S}^2&=&4\,\one\otimes\one\,-\,[\hat{A}_0,\hat{A}_1]\otimes [\hat{B}_0,\hat{B}_1]\,.
\ea
Indeed,
\ban
||[\hat{A}_0,\hat{A}_1]||_{\infty}&=&||\hat{A}_0\hat{A}_1-\hat{A}_1\hat{A}_0||_{\infty}\,\leq\,||\hat{A}_0\hat{A}_1||_{\infty}+||\hat{A}_1\hat{A}_0||_{\infty}\,\leq 2||\hat{A}_0||_{\infty}\,||\hat{A}_1||_{\infty} \,=\,2
\ean where for the last estimate we have used the inequality $|xy|\leq |x|\,|y|$. The same estimate leads to $||[\hat{B}_0,\hat{B}_1]||_{\infty}\leq 2$. Therefore $||\hat{S}^2||_\infty\leq 8$, which proves the claim.\end{proof}

\subsection{Study of CHSH for two-qubit states}
\label{sschshqubits}

Historically, the first state for which a violation of Bell inequalities was noticed is the singlet state $\ket{\Psi^-}=\frac{1}{\sqrt{2}}(\ket{0}\otimes\ket{1}-\ket{1}\otimes\ket{0})$. As mentioned above, for this state quantum theory predicts
\ba
E_{\vec{a},\vec{b}}(\Psi^-)&=&-\vec{a}\cdot\vec{b}
\ea
for Alice measuring along direction $\vec{a}$, and Bob along direction $\vec{b}$, in the Bloch sphere. With suitable choices of the measurements, these statistics can reach the Tsirelson bound $S=2\sqrt{2}$. Instead of indulging on this specific, very well known case, I provide directly a characterization of the behvaior of CHSH for a generic two-qubit state
\ba
\rho&=&\frac{1}{4}\Big(\one\otimes\one+\vec{r}_\rho\cdot\vec{\sigma}\otimes\one+\one\otimes\vec{s}_\rho\cdot\vec{\sigma}+\sum_{i,j=x,y,z}T_\rho^{ij}\sigma_i\otimes\sigma_j\Big)\,.\label{genqubit}
\ea

\begin{theorem}\label{thm2} The maximal value of CHSH achievable with von Neumann measurements on a generic two-qubit state $\rho$ \eqref{genqubit} is
\ba
\moy{\hat{S}}&=&2\sqrt{\lambda_1+\lambda_2}\label{maxviolqubitgen}
\ea where $\lambda_1$ and $\lambda_2$ are the two largest eigenvalues of the symmetric matrix $T^t_\rho T_\rho$, where $T^t_\rho$ is the transpose of $T_\rho$. The proof contains the construction of a possible choice of settings that reaches this maximum \cite{horo3}.
\end{theorem}

\begin{proof} The quantity to be maximized is
\ba
\moy{\hat{S}}\,=\,\mathrm{Tr}\big(\rho \hat{S}\big)&=&\Big(\vec{a}_0,T_\rho\underbrace{(\vec{b}_0+\vec{b}_1)}_{=2\cos\chi\vec{c}}\Big)\,+\,\Big(\vec{a}_1,T_\rho\underbrace{(\vec{b}_0-\vec{b}_1)}_{=2\sin\chi\vec{c}^\perp}\Big)
\ea where we used the fact that the sum and difference of two vectors are always orthogonal vectors, and where $(\vec{b}_0,\vec{b}_1)=\cos2\chi$.

Since $||\vec{a}_0||=1$, the maximal value of $\big(\vec{a}_0,T_\rho\vec{c}\big)$ is $||T_\rho\vec{c}||$, obtained when $\vec{a}_0$ is chosen parallel to $T_\rho\vec{c}$. By the same reasoning on the term involving $\vec{a}_1$, we reach
\ba
\max_{\vec{a}_0,\vec{a}_1,\vec{b}_0,\vec{b}_1}\moy{\hat{S}}&=&\max_{\vec{b}_0,\vec{b}_1}2\big(\cos\chi ||T_\rho \vec{c}|| + \sin\chi ||T_\rho \vec{c}^\perp||\big) \nonumber \\&=& \max_{\vec{c},\vec{c}^\perp}2\sqrt{||T_\rho \vec{c}||^2 + ||T_\rho \vec{c}^\perp||^2}\label{maxqubinterm}
\ea where in the last step we used the well known optimization $\max_{\chi}\cos\chi x + \sin\chi y = \sqrt{x^2+y^2}$ achievable with the choice $\cos\chi=x/\sqrt{x^2+y^2}$. Finally, $||T_\rho \vec{c}||^2=\Big(\vec{c},T_\rho^t T_\rho\vec{c}\Big)$ and $T_\rho^t T_\rho$ is a positive symmetric matrix. Therefore the maximization \eqref{maxqubinterm} is achieved by choosing $\vec{c}$ and $\vec{c}^\perp$ as the two eigenvectors associated to the two largest eigenvalues. This proves \eqref{maxviolqubitgen}. The proof gives the recipes to reconstruct the measurement settings that lead to the maximal violation.

\end{proof}

Let us particularize this result to the case of a pure state:
\ba
\ket{\Psi(\theta)}=\cos\theta\,\ket{0}\otimes\ket{0}+\sin\theta\, \ket{1}\otimes\ket{1} &\longrightarrow&
T_{\Psi(\theta)}=\left(\begin{array}{ccc}\sin 2\theta &0&0\\ 0&-\sin 2\theta& 0\\ 0& 0& 1 \end{array}\right)\,.
\ea
The maximal achievable value of CHSH is
\ba
\max_{\vec{a}_0,\vec{a}_1,\vec{b}_0,\vec{b}_1}\sandwich{\Psi(\theta)}{\hat{S}}{\Psi(\theta)} &=& 2\,\sqrt{1+\sin^22\theta}\,,\label{maxviolqubitpure}
\ea which is always larger than 2, unless $\sin 2\theta=0$. Therefore, all pure entangled states of two qubits violate CHSH for suitable measurements. Moreover, only maximally entangled states can reach $S=2\sqrt{2}$.

Let us determine the corresponding measurement settings. The eigenvector associated to the largest eigenvalue of $T_\rho^t T_\rho$ is $\vec{c}=\hat{z}$; the orthogonal subspace being degenerate, we can choose any vector in it as $\vec{c}^\perp$: for instance $\vec{c}^\perp=\hat{x}$. With this choice,
\ba
\vec{b}_{0,1}\,=\,\cos\chi\,\hat{z}\,\pm\,\sin\chi\,\hat{x} &\textrm{ with }&\cos\chi\,=\,1/\sqrt{1+\sin^22\theta}\,.
\ea
Furthermore, $\vec{a}_0$ must be the unit vector parallel to $T_\rho\vec{c}$, and $\vec{a}_1$ to $T_\rho\vec{c}^\perp$, so here
\ba
\vec{a}_0\,=\,\hat{z} &,& \vec{a}_1\,=\,\hat{x}\,.
\ea
Now we know pretty much everything about the violation of CHSH by two-qubit states. In the next paragraph, we are going to see how CHSH can be adapted to prove that all pure entangled states violate a Bell inequality.

\subsection{All pure entangled states violate a Bell inequality (``Gisin's theorem")}

The fact that \textit{all pure entangled states violate a Bell inequality} is usually referred to as Gisin's theorem, since Nicolas Gisin was the first to ask the question and to answer it for bipartite states \cite{gisinthm}. Popescu and Rohrlich extended the proof to the general case shortly after \cite{prgisthm}. Here, I follow this development. 

\begin{lemma} Any \textit{bipartite} pure state (i.e.~of any dimensionality) violates a Bell inequality.\label{thm:gisin1}
\end{lemma}

\begin{proof} Any bipartite pure state can be written in its Schmidt decomposition
$\ket{\Psi}=\sum_{k=0}^{d-1} c_k \ket{k}\otimes\ket{k}$ where we can define the bases such that $c_0\geq c_1\geq ... \geq c_{d-1}\geq 0$. The state is entangled if and only if $c_1\neq 0$. Let us now rewrite the state as
\ba
\ket{\Psi}&=&\sqrt{c_0^2+c_1^2}\,\big(\cos\theta\ket{0}\otimes\ket{0}+\sin\theta \ket{1}\otimes\ket{1}\big)\,+\,\sqrt{1-c_0^2-c_1^2}\,\ket{\Psi'}
\ea where $\cos\theta=c_0/\sqrt{c_0^2+c_1^2}$ and $\ket{\Psi'}$ is the normalized projection of $\ket{\Psi}$ onto the subspace orthogonal to $\textrm{Span}(\ket{0}\otimes\ket{0},\ket{1}\otimes\ket{1})$. Now consider the operators
\ba
A_x\,=\,\vec{a}_x\cdot\vec{\sigma} \oplus \one'&,& B_y\,=\,\vec{b}_y\cdot\vec{\sigma} \oplus \one'
\ea where for both systems the Pauli matrices are defined as acting in the subspace $\textrm{Span}(\ket{0},\ket{1})$ and $\one'=\ket{2}\bra{2}+...+\ket{d-1}\bra{d-1}$ is the identity on the orthogonal complement of that subspace. By choosing the measurement vectors that lead to \eqref{maxviolqubitpure}, one can therefore reach a value of CHSH
\ba
S&=&(c_0^2+c_1^2)\,2\sqrt{1+\sin2\theta}\,+\,(1-c_0^2-c_1^2)\,2
\ea which is larger than 2 as soon as $c_1>0$ as claimed.
\end{proof}
There is no claim that this procedure is optimal according to any figure of merit: after all, the construction does not make any use of the entanglement that was possibly present outside $\textrm{Span}(\ket{0}\otimes\ket{0},\ket{1}\otimes\ket{1})$. One may want to look for ``better" inequalities, even some that are tailored to the state. However, this approach is sufficient to prove the claim; moreover, it has the advantage that the procedure is defined uniquely, once and for all. This same remark will apply also to the general theorem, which can now be proved:

\begin{theorem}\label{gisthmfull} Any pure state violates a Bell inequality, and this can be checked by a uniquely defined protocol in which each of the parties performs local measurements\footnote{Obviously, if the parties could come together and form two groups, one goes back to the case settled by Lemma \ref{thm:gisin1}.}.
\end{theorem}

\begin{proof} It is better to introduce the idea with an example. Suppose that Alice, Bob and Charlie share the GHZ state of three qubits $\frac{1}{\sqrt{2}}(\ket{000}+\ket{111})$ and they have to perform local measurements, but they have only heard of the CHSH inequality. They can do the following: Charlie measures only $\sigma_x$; if he finds $+1$, he has prepared the state $\ket{\Phi^+}=\frac{1}{\sqrt{2}}(\ket{00}+ \ket{11})$ for Alice and Bob; otherwise, he has prepared the state $\ket{\Phi^-}$. Alice and Bob, as for them, alternate between two measurement settings, those suitable to violate CHSH with $\ket{\Phi^+}$. Clearly, the statistics $P(a,b|x,y,z=\sigma_x,c=+1)$ will exhibit $S=2\sqrt{2}$. Now, it is not difficult to convince oneself that, if the $P(a,b,c|x,y,z)$ can be described by a LV model, then all the $P(a,b|x,y,z,c)$ can be described by a LV model too. Therefore \textit{a contrario}, if the $P(a,b|x,y,z,c)$ violate a Bell inequality for some choices of $z$ and $c$ (as in the example), then the original three partite statistics $P(a,b,c|x,y,z)$ cannot be described by LV\footnote{In order for Alice and Bob to check the violation, at some point Charlie has to send $(z,c)$ to them. This is not problematic, as can be seen in two scenarios: (i) If the information is sent before Alice and Bob choose their measurements, it can be seen as part of the procedure of state preparation, i.e.~part of the pre-established agreement; (ii) If the information is sent at the end of the protocol, the post-selection done by Alice and Bob does not open the detection loophole, because it is uncorrelated with their choice of settings.}.

In general, if a pure multipartite state is entangled, at least one pair can be prepared in a bipartite pure entangled state $\ket{\Psi}_{AB}$ by the other parties $C_1,C_2,...,C_N$ performing suitable measurements and obtaining the right outcomes. Given the state, therefore, the $N+2$ parties agree that $N$ of them perform only a single measurement, while the remaining two parties choose suitable measurements for $\ket{\Psi}$ to violate some Bell inequality. Lemma \ref{thm:gisin1} guarantees that such measurements always exist. The violation implies that $P(a,b,c_1,...,c_N|x,y,z_1,...,z_N)$ cannot be described by a LV model. \end{proof}

\subsection{Some mixed entangled states don't violate any Bell inequality: Werner states}
\label{sswerner}

It is pretty obvious that \textit{entanglement is necessary to violate Bell's inequalities in quantum theory} and all separable states admit a LV model. It is very reasonable to conjecture that entanglement is also sufficient to violate Bell's inequalities --- but this conjecture is wrong, as first proved by Werner  \cite{wer}:
\begin{theorem}\label{thm:werner}
There exist mixed states that are entangled, but nevertheless cannot violate any Bell inequality.
\end{theorem}

The proof will be given by exhibiting explicit counterexamples. Consider the single-parameter family of two-qubit states (``Werner states") defined as
\ba
\rho_W\,=\,W\ket{\Psi^-}\bra{\Psi^-}\,+\,(1-W)\frac{\one}{4} &,&W\in [0,1]\,.
\ea Using the criterion of the negative partial transposition, it can be proved that Werner states are separable for $W\leq \frac{1}{3}$ and entangled otherwise. The statistics of von Neumann measurements on Werner states are given by
\ba
{\cal P}_{Q}(W)&=&\left\{P(a,b|\vec{a},\vec{b})=\frac{1}{4}(1-W ab\,\vec{a}\cdot\vec{b})\,,\,a,b\in\{-1,+1\}, \vec{a},\vec{b}\in \mathbb{S}^2\right\}\,.
\label{statswerner}\ea

\begin{lemma}
The set ${\cal P}_{Q}(W)$ can be reproduced with LV if $W\leq \demi$.
\end{lemma}

\begin{proof} It is clearly sufficient to exhibit the proof for $W=\demi$, since any $\rho_W$ with $W<\demi$ can be obtained by mixing $\rho_{\demi}$ with white noise. 

In each run, the pre-shared local variable is a vector $\vec{\lambda}$ drawn from the unit sphere $\mathbb{S}^2$ with uniform distribution  $\rho(\vec{\lambda})d\vec{\lambda} = \frac{1}{4\pi}\sin\theta d\theta d\varphi$ with the usual spherical coordinates. Alice's box simulates the measurement of a single spin prepared in the direction $\vec{\lambda}$:
\ba
P^A_{\vec{\lambda}}(a|\vec{a})&=&\demi\left(1+a\,\vec{a}\cdot\vec{\lambda}\right)\,.
\ea Bob's box outputs $b=-\textrm{sign}(\vec{b}\cdot\vec{\lambda})$, i.e. $b=+1$ if $\vec{b}\cdot\vec{\lambda}\leq 0$, $b=-1$ if $\vec{b}\cdot\vec{\lambda}> 0$). So we have
\ba
P(a,+1|\vec{a},\vec{b})&=&\int_{S^2}d\vec{\lambda}\rho(\vec{\lambda})\,P^A_{\vec{\lambda}}(a|\vec{a})\,\delta_{\vec{b}\cdot\vec{\lambda}\leq 0}\,=\,\frac{1}{4}+\frac{1}{2}\,a\,\int_{\vec{b}\cdot\vec{\lambda}\leq 0}d\vec{\lambda}\rho(\vec{\lambda})\,\vec{a}\cdot\vec{\lambda}\,.\label{wercomp}
\ea In order to compute the integral, we choose the spherical coordinates such that $\vec{b}$ is $\hat{z}$ (i.e. $\theta=0$):
\ban
\int_{\vec{b}\cdot\vec{\lambda}\leq 0}d\vec{\lambda}\rho(\vec{\lambda})\,\vec{a}\cdot\vec{\lambda}&=& \frac{1}{4\pi}\int_{\pi/2}^{\pi}d\theta\sin\theta  \int_{0}^{2\pi}d\varphi \Big[(a_x\cos\varphi+a_y\sin\varphi)\sin\theta +a_z\cos\theta\Big]\nonumber\\&=&\demi\,a_z\,\underbrace{\int_{\pi/2}^{\pi}d\theta\sin\theta \cos\theta}_{=-\demi}=-\frac{1}{4}a_z\,.
\ean Inserting this result into \eqref{wercomp} and recalling that $a_z=\vec{a}\cdot\vec{b}$, we recover indeed \eqref{statswerner} for $b=+1$. The calculation for $b=-1$ changes only in the bounds of the last integral and yields the desired result too. This concludes the proof of the Lemma. Since all the Werner states with $\frac{1}{3}<W\leq\demi$ are entangled, it proves Theorem \ref{thm:werner} as well.\end{proof}

\subsection{All that I left out}

The goal of this text is to present the device-independent outlook. In this perspective, the cases of violation of Bell inequalities in quantum physics constitute a \textit{bank of examples}. I have kept this section short because we have collected sufficiently many examples to move on. But I encourage the reader to learn more examples by referring to the comprehensive review \cite{ourreview}.

In particular, the reader may want to learn about \textit{other Bell inequalities}: with more inputs, more outputs, more parties. One will learn, first of all, that the computation of Tsirelson bounds and of violations for a family of states are not obvious tasks, contrary to what our specific Theorems \ref{thm1} and \ref{thm2} may suggest. Intriguing features emerge, which the simple CHSH does not possess: for instance, some inequalities are surprisingly maximally violated by non-maximally entangled states; other inequalities can be used as dimension witnesses, because the maximal violation cannot be achieved with (say) two qubits and requires in fact quantum systems of dimension larger than some bound.

It is also important to realize that the ultimate test of the non-classicality of entangled states need not consist of the elementary procedure ``take a state and measure it". The Bell test may be the final step of \textit{more complicated procedures}, involving for instance an initial filtering (\textit{hidden non-locality}) or multi-partite preparation (\textit{activation of non-locality}).

\section{Device-independent assessment}
\label{secdi}

\subsection{A new vantage point}

When I entered the field back in the year 2000, Bell inequalities seemed to have already fulfilled their mission. Sure enough, a few passionate researchers were still distilling moderately interesting mathematical and physical insight from them. But most colleagues regarded them as one would regard, for instance, Foucault's pendulum: an instrument, which has allowed humankind to firmly establish a crucial fact about our physical world. Both the pendulum and the inequalities should feature in the science museums of the world, so that every curious human would be informed that the Earth rotates around its axis and that there exists intrinsic randomness --- but research should move forward. For a few years, there was no strong argument against this stance: the few of us, who were still researching on Bell inequalities and related notions, looked like the lovers of Kodak films enjoying their last moments of fun at the dawn of the digital era.

It all changed thanks to a complex chain of thoughts, stretching between 2005 and 2007, which I narrate in Appendix \ref{sspathdi} because I have chosen to order the materials in this text in a logical, rather than chronological, sequence. The final outcome can however be explained here because it was, in hindsight, a \textit{flash of the obvious}.

I stressed enough many times that Bell inequalities are independent of quantum theory --- and so they must be, if they have to test quantum theory against the alternative description of pre-established agreement. The flash of the obvious is that one can be more specific: \textit{the assessment of a Bell test does not rely on the knowledge of the degrees of freedom that are measured}. Notions like ``photons", ``polarization", ``complementary bases", ``dimension of the Hilbert space", are completely dispensed with at the moment of \textit{analyzing} a Bell test (though, needless to say, the experimentalists who \textit{build} the test had better have a very good control of all that). In other words, Bell inequalities are the only entanglement witnesses that do not rely on assumptions about dimensionality of the state or commutation relations of the observables.

This observation opens up the possibility of \textit{device-independent assessment} of entanglement and of related operational quantities like the secrecy of a cryptographic key, the amount of intrinsic randomness that is generated, etc. Besides their undisputed role in shaping the scientific worldview, Bell inequalities have a very unique role to play in future quantum technologies by providing the \textit{ultimate level of device certification}.

\subsection{Self-testing}

Self-testing, which could also be called \textit{device-independent characterization of the state and the measurements}, or simply \textit{blind tomography}, refers to the fact that some statistics predicted by quantum theory determine the state and the measurement as uniquely as possible, namely up to a local isometry. We first explain this equivalence class, then discuss explicit examples of self-testing.

\subsubsection{Equivalence up to a local isometry}

If the observed statistics $P(a,b|x,y)$ are the only available data, the state and the measurements can be characterized at most \textit{up to local isometries}. Indeed, if nothing is assumed about the \textit{dimension}, one obtains the same statistics \ba
P(a,b|x,y)&=&\textrm{Tr}(\rho\,E_a^x\otimes E_b^y)=\textrm{Tr}(\tilde{\rho}\,\tilde{E}_a^x\otimes \tilde{E}_b^y)\ea
by appending other degrees of freedom according to $\tilde{\rho}=\rho\otimes \sigma_{A'B'C}$, and performing trivial measurements on those ($\tilde{E}_a^x=E_a^x\otimes \one_{A'C}$ and $\tilde{E}_b^y=E_b^y\otimes\one_{B'C}$). Further, if nothing is assumed about the \textit{measurements}, one obtains the same statistics if states and measurements differ by local unitaries, i.e. $\tilde{\rho}=U_A\otimes U_B\rho U_A^\dagger\otimes U_B^\dagger$, $\tilde{E}_a^x=U_A E_a^x U_A^\dagger$ and $\tilde{E}_b^y=U_B E_b^y U_B^\dagger$.

Let us practice this notion on an example. Consider the state
\ba
\ket{\Psi}&=&\sum_{k=0,1,...} c_{k}\,\frac{\ket{2k,2k}+\ket{2k+1,2k+1}}{\sqrt{2}}\,.\label{selftest0}
\ea This is visibly a direct sum of singlets --- since local unitaries are free, here I shall call ``singlet" any maximally entangled state of two qubits and shall rather use $\ket{\Phi^+}=\frac{\ket{00}+\ket{11}}{\sqrt{2}}$. One expects this state to exhibit at least all the properties of the singlet. A local isometry helps to make this fact manifest. Let us append to $\ket{\Psi}_{AB}$ a two-qubit local ancilla $\ket{00}_{A'B'}$ and apply to both sides the local isometry $\Phi=\Phi_A\otimes\Phi_B$ with
\begin{subequations}
\begin{align}
\Phi_A \ket{2k,0}_{AA'} & \mapsto \ket{2k,0}_{AA'}, \\ \Phi_A \ket{2k+1,0}_{AA'} & \mapsto \ket{2k,1}_{AA'}, 
\end{align}
\end{subequations} and $\Phi_B$ identically defined. We obtain
\ba
\ket{\Psi}_{AB}\ket{00}_{A'B'}& \mapsto & \Big[c_0\ket{00}+c_1\ket{22}+...+c_k\ket{2k,2k}+...\Big]_{AB}\ket{\Phi^+}_{A'B'}\,.
\ea
The local isometry has mapped the singlet into $A'B'$, while $AB$ carry the rest of the structure (in which there may be a lot of entanglement left).

With these notions into place, we can study self-testing.

\subsubsection{Self-test of the singlet using CHSH}

The simplest self-test criterion to state is probably the following:
\begin{theorem}
If a CHSH test yields $S=2\sqrt{2}$ exactly, then, up to local isometries, the state is a singlet of two qubits and the measurements are the corresponding Pauli matrices.
\end{theorem}
In other words, any state that leads to $S=2\sqrt{2}$ can be written as \eqref{selftest0}, or as a mixture of such states out of which the singlet can be extracted with the same isometry.

This theorem has been proved in various ways \cite{pr92,bmr92}, and ultimately the proof can be related to the one given below; so I skip it here. But two remarks are worth making. First, it is remarkable that self-testing can be achieved from a single number, associated to a procedure that uses only two measurements per site. Second, only extremal points of the set of quantum statistics can be self-tested exactly, but these points can never be achieved in a real experiment: therefore, self-testing calls for robustness bounds if it is ever to become useful. As it turns out, such bounds have been given; but I'll skip them, since they are not optimal and only add technicalities to the proofs.

\subsubsection{Self-test of the singlet using the Mayers-Yao statistics}

Consider two settings labeled $\{X_A,Z_A\}$ on Alice's side and three settings labeled $\{X_B,Z_B,D_B\}$ on Bob's side. All the measurements are binary and their outcomes are labeled $\pm 1$. The locality of the measurements $[M_A,N_B]=0$ is assumed; apart from this, nothing is known \textit{a priori} about these measurements. In particular, the dimensionality is not known: this implies that any experiment can be described by the \textit{projective} measurement of an unknown but \textit{pure} bipartite state $\ket{\Psi}$. In turn, this implies that for each of the measurements being performed, there exist two projectors $\Pi_{+1}$ and $\Pi_{-1}$ associated to the two outcomes. The pure state and the two projectors are the mathematical objects that describe the content of the boxes; all the rest will be constructed from them. In particular, for each measurement we construct the operator $M=\Pi_{+1}-\Pi_{-1}$, which by construction is unitary and hermitian, with the property $M^2=\one$. The five operators $\{X_A,Z_A;X_B,Z_B,D_B\}$ must be considered as mathematical objects built in this way\footnote{This is crucial for understanding the proof. Notably, $M^2$ does \textit{not} refer to two sequential application of the unknown box (if this were the case, there would be no guarantee that $M^2=\one$ without having to assume that the boxes are non-demolition L\"uders instruments). In each run, the unknown box is used only once. The rest are mathematical manipulations based on the fact that each box must define those two projectors.}.

Suppose now that one observes that all the $\sandwich{\Psi}{M_AN_B}{\Psi}$ are equal to the two-qubit values one would obtain for the singlet and the corresponding Pauli matrices $\sandwich{\Phi^+}{\sigma_m\otimes\sigma_n}{\Phi^+}$, with the identification $\sigma_d=\frac{1}{\sqrt{2}}(\sigma_z+\sigma_x)$: then $\ket{\Psi}$ is equivalent to $\ket{\Phi^+}$ up to local isometries, and moreover all the operators are also equivalent to Pauli matrices acting on the effective qubit. This is a modification of the original Mayers-Yao scheme \cite{my04}, which used three settings also on Alice's side; the proof can be presented here thanks to the work of Matthew McKague, who found a way of rederiving the original result in simple terms.

\begin{theorem}
Consider five unknown unitary operators $\{X_A,Z_A;X_B,Z_B,D_B\}$ with binary outcomes labeled $\pm 1$ and assumed to fulfill $[M_A,N_B]=0$: if
\ba
\sandwich{\Psi}{Z_AZ_B}{\Psi}\,=\,\sandwich{\Psi}{X_AX_B}{\Psi}&=&1\label{selftest1}\\
\sandwich{\Psi}{X_AZ_B}{\Psi}\,=\,\sandwich{\Psi}{Z_AX_B}{\Psi}&=&0\label{selftest2}\\
\sandwich{\Psi}{Z_AD_B}{\Psi}\,=\,\sandwich{\Psi}{X_AD_B}{\Psi}&=&1/\sqrt{2}\label{selftest3}
\ea then there exist a local isometry $\Phi=\Phi_A\otimes\Phi_B$ such that
\ba
\Phi\ket{\Psi}_{AB}\ket{00}_{A'B'}&=&\ket{\mathrm{junk}}_{AB}\ket{\Phi^+}_{A'B'}\,,\\
\Phi\,M_AN_B\ket{\Psi}_{AB}\ket{00}_{A'B'}&=&\ket{\mathrm{junk}}_{AB}\left(\sigma_m\otimes\sigma_n\ket{\Phi^+}_{A'B'}\right)\,.
\ea
\end{theorem}

\begin{proof} The condition \eqref{selftest1} can be rewritten as
\ba
Z_A\ket{\Psi}=Z_B\ket{\Psi}\textrm{ and }X_A\ket{\Psi}=X_B\ket{\Psi}\label{st1}\,.
\ea
Inserting this into \eqref{selftest2}, we find $\sandwich{\Psi}{X_AZ_A}{\Psi}=0$, that is $X_{A}\ket{\Psi}$ is orthogonal to $Z_{A}\ket{\Psi}$. If this is the case, then \eqref{selftest3} means that 
\ba D_B\ket{\Psi}&=&\frac{Z_A+X_A}{\sqrt{2}}\ket{\Psi}\,:\label{st3}\ea indeed, $D_B\ket{\Psi}$ must be of norm 1, and we know already two of its projections of amplitude $\frac{1}{\sqrt{2}}$ on two orthogonal vectors, so there can't be anything left. The last preparatory step consists in computing $D_B^2\ket{\Psi}$. One must be slightly careful here, because \eqref{st3} tells us how $D_B$ behaves on $\ket{\Psi}$, not how it behaves on the different vector $D_B\ket{\Psi}$. Nevertheless, using $[M_A,N_B]=0$ we obtain
\ban
D_B^2\ket{\Psi}&=&D_B\,\frac{Z_A+X_A}{\sqrt{2}}\ket{\Psi}\,=\,\frac{Z_A+X_A}{\sqrt{2}}\,D_B\ket{\Psi}\,=\,\Big(\frac{Z_A+X_A}{\sqrt{2}}\Big)^2\ket{\Psi}\\
&=&\demi(Z_A^2+X_A^2+Z_AX_A+X_AZ_A)\ket{\Psi}\,.
\ean
But $D_B^2=Z_A^2=X_A^2=\one$, so finally
\ba
Z_AX_A\ket{\Psi}\,=\,-X_AZ_A\ket{\Psi}&\stackrel{\eqref{st1}}{\Longrightarrow}& Z_BX_B\ket{\Psi}\,=\,-X_BZ_B\ket{\Psi}\,.\label{st4}
\ea
The third setting $D_B$ was instrumental in deriving these anti-commutation relations and has now finished its role.

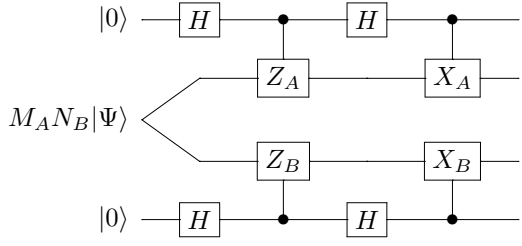
\begin{figure}[t]
\[
\Qcircuit @C=0.5cm @R=.3cm {
\lstick{\ket{0}}  & \gate{H} & \ctrl{1}  & \gate{H} & \ctrl{1} &\qw \\
 &  & \gate{Z_{A}} & \qw & \gate{X_{A}} & \qw\\
 \lstick{M_{A} N_{B}\ket{\Psi}} \ar@{-}[ur] \ar@{-}[dr] & & & & & \\
              &    & \gate{Z_{B}} & \qw & \gate{X_{B}} & \qw\\
\lstick{\ket{0}}     & \gate{H} & \ctrl{-1} & \gate{H} & \ctrl{-1} & \qw
}
\]
\caption{Local isometry $\Phi$ that allows self-testing of the singlet and all the measurements of the Mayers-Yao test ($M,N\in \{I,X,Z\}$). From top to bottom, the rows represent the systems $A'$ (qubit), $A$ (unconstrained), $B$ (unconstrained) and $B'$ (qubit). $H$ is the Hadamard unitary gate defined as usual: $H\ket{0}=\frac{1}{\sqrt{2}}(\ket{0}+\ket{1})$, $H\ket{1}=\frac{1}{\sqrt{2}}(\ket{0}-\ket{1})$. The controlled gates act non-trivially on the target when the control qubit is in state $\ket{1}$.}
\label{figisom}
\end{figure}

The rest of the proof consists in exhibiting an explicit local isometry which leads to self-test of the singlet and of the corresponding measurements when \eqref{st1} and \eqref{st4} hold\footnote{The proof of self-testing from CHSH can be made by showing that $S=2\sqrt{2}$ implies \eqref{st1} and \eqref{st4} for suitably defined operators; then continuing with the same steps as we are going to describe.}. The isometry is described in Fig.~\ref{figisom}: a simple step-by-step calculation shows that it implements the transformation
\ba\label{transfstate}
\Phi\ket{\Psi}_{AB}\ket{00}_{A'B'}&=&\frac{1}{4}\Big[[(\one+Z_A)(\one+Z_B)\ket{\Psi}]\,\ket{00}\nonumber\\&&+ [X_AX_B(\one-Z_A)(\one-Z_B)\ket{\Psi}]\,\ket{11}\nonumber\\
&&+[X_B(\one+Z_A)(\one-Z_B)\ket{\Psi}]\,\ket{01}\nonumber\\&&+[X_A(\one-Z_A)(\one+Z_B)\ket{\Psi}]\,\ket{10}\Big]\,.
\ea
First, using \eqref{st1}, we replace $Z_B$ with $Z_A$, and this cancels the third and fourth lines because $(\one+Z_A)(\one-Z_A)=\one-Z_A^2=0$. Then, using \eqref{st4}, one proves that $X_AX_B(\one-Z_A)(\one-Z_B)\ket{\Psi}=(\one+Z_A)(\one+Z_B)X_AX_B\ket{\Psi}$, which is in turn equal to $(\one+Z_A)(\one+Z_B)\ket{\Psi}$ because of \eqref{st1}. Finally $(\one+Z_A)(\one+Z_B)\ket{\Psi}=(\one+Z_A)^2\ket{\Psi}=2(\one+Z_A)\ket{\Psi}$, so we have found
\ban
\Phi\ket{\Psi}_{AB}\ket{00}_{A'B'}&=&\frac{1}{2}(\one+Z_A)\ket{\Psi}\,[\ket{00}+\ket{11}]\,=\,
\underbrace{\frac{\one+Z_A}{\sqrt{2}}\ket{\Psi}_{AB}}_{=\ket{\mathrm{junk}}_{AB}}\,\ket{\Phi^+}_{A'B'}
\ean which is the self-testing of the state.

The proof for the measurements follows the same steps, starting with $\Phi\,M_AN_B\ket{\Psi}_{AB}\ket{00}_{A'B'}$ instead of \eqref{transfstate}. Let me show it for one of the six cases:
\ban
\Phi\,\underline{X_A}\ket{\Psi}_{AB}\ket{00}_{A'B'}&=&\frac{1}{4}\Big[[(\one+Z_A)(\one+Z_B)\underline{X_A}\ket{\Psi}]\,\ket{00}\nonumber\\&&+ [X_AX_B(\one-Z_A)(\one-Z_B)\underline{X_A}\ket{\Psi}]\,\ket{11}\nonumber\\
&&+[X_B(\one+Z_A)(\one-Z_B)\underline{X_A}\ket{\Psi}]\,\ket{01}\nonumber\\&&+[X_A(\one-Z_A)(\one+Z_B)\underline{X_A}\ket{\Psi}]\,\ket{10}\Big]
\ean By using \eqref{st4}, one moves $X_A$ to the left in the first, second and fourth lines while changing the sign of $Z_A$, and $X_B$ to the right in the third while changing the sign of $Z_B$. The analysis is then the same as above, and the result is
\ban
\Phi\,X_A\ket{\Psi}_{AB}\ket{00}_{A'B'}&=& \frac{1}{2}(\one+Z_A)\ket{\Psi}\,[\ket{01}+\ket{10}]\,=\,
\ket{\mathrm{junk}}_{AB}\,\left(\sigma_x\otimes\one\,\ket{\Phi^+}_{A'B'}\right)\,.
\ean This concludes the proof, which was unpublished as such, but is a particular case of the robustness proofs presented in \cite{mys12}. \end{proof}


\subsection{Intermezzo: two tools}

The main difficulty that one encounters in device-independent studies is the fact that there are \textit{a priori} infinitely many free parameters, since the dimension of the Hilbert space is arbitrary. A reduction to a finite problem seems to be necessary if one wants to compute explicit bounds. In self-testing, for instance, the isometry takes care of relegating all the potential infinity into the junk --- by the way, in the text I have by-passed the main difficulty faced in the research work by providing the isometry, instead of leaving it to be found by the reader. Before moving to the next example of device-independent assessment, we need to introduce two of the main other tools developed so far. They are not on equal footing: the first one is very specific to the CHSH case, the second is far more general.

\subsubsection{A very specific tool: decomposition of CHSH}
\label{sssblocks}

We start with Jordan's lemma:
\begin{lemma}\label{lemmabinary}
Let $\hat{A}_0$ and $\hat{A}_1$ be two Hermitian operators with eigenvalues $-1$ and $+1$. There exist a basis in which both operators are block-diagonal, in blocks of dimension $2\times 2$ at most.
\end{lemma}

\begin{proof} By definition, $\hat{A}_0^2=\hat{A}_1^2=\one$ in the suitable (unknown) dimension. It is then trivial to check that $U=\hat{A}_0\hat{A}_1$ is unitary. Let us denote $\ket{\alpha,0}$ an eigenstate of $U$: $U\ket{\alpha,0}=\omega_\alpha\ket{\alpha,0}$ with $|\omega_\alpha|=1$. Then $\ket{\alpha,1}=\hat{A}_0\ket{\alpha,0}$ is also an eigenstate of $U$: indeed, $U\ket{\alpha,1}=\hat{A}_0\hat{A}_1\hat{A}_0\ket{\alpha,0}=\hat{A}_0U^\dagger\ket{\alpha,0}=\omega_\alpha^*\ket{\alpha,1}$. Therefore:
\ba
\hat{A}_0\ket{\alpha,0}=\ket{\alpha,1}&\mathrm{ and }&\hat{A}_0\ket{\alpha,1}=\ket{\alpha,0}\,,\\
\hat{A}_1\ket{\alpha,1}=U^\dagger\ket{\alpha,0}=\omega_\alpha^*\ket{\alpha,0} &\mathrm{ and }& \hat{A}_1\ket{\alpha,0}=\hat{A}_1(\omega_\alpha\hat{A}_1\ket{\alpha,1})=\omega_\alpha\ket{\alpha,1}\,.
\ea Since the eigenvectors of a unitary operator span the whole Hilbert space, we have
\ba
\hat{A}_0=\bigoplus_{\alpha} \sigma_x^\alpha &\mathrm{and}& \hat{A}_1=\bigoplus_{\alpha} \mathrm{Re}(\omega_\alpha)\,\sigma_x^\alpha-\mathrm{Im}(\omega_\alpha)\,\sigma_y^\alpha
\ea where the $\sigma_k^\alpha$ are the usual Pauli matrices defined in $\textrm{Span}\{\ket{\alpha,0},\ket{\alpha,1}\}$ with the conventional choice $\sigma_z^\alpha=\ket{\alpha,0}\bra{\alpha,0}-\ket{\alpha,1}\bra{\alpha,1}$.
\end{proof}

As a direct consequence of this Lemma applied to both Alice's and Bob's boxes, there exist a basis in which the CHSH operator can be decomposed as
\ba\label{decompchsh}
\hat{S}&=&\bigoplus_{\alpha} \bigoplus_{\beta} \hat{S}^{\alpha\beta}
\ea
where each of the $\hat{S}^{\alpha\beta}$ is a two-qubit CHSH operator. In other words, any observed violation of CHSH can be written as \ba
S_{obs}&=&\Tr(\rho S)\,=\,\sum_{\alpha,\beta} \Tr(\rho^{\alpha\beta} S^{\alpha\beta})\label{sobsblocks}
\ea where $\rho^{\alpha\beta}$ is the unnormalized two-qubit state whose matrix elements are
\ban
\rho_{2i+j+1,2k+l+1}^{\alpha\beta}&=&\sandwich{\alpha,i;\beta,j}{\rho}{\alpha,k;\beta,l}\;,\;\;i,j,k,l\in\{0,1\}\,.
\ean
One can therefore hope to reduce a device-independent study of CHSH to a two-qubit problem. This reduction does not follow automatically from the decomposition\footnote{There is still quite some complexity left in \eqref{decompchsh}: for instance, one cannot optimize each $\hat{S}^{\alpha\beta}$ independently, since the Pauli matrices of Alice (Bob) are the same in all operators with the same $\alpha$ ($\beta$). Such constraints are notably the reason why the conjectured maximal violation of CHSH by higher dimensional pure states is still unproven (see \cite{liang} for the most thorough numerical verification). Another example, in which the reduction does not seem to be possible, is the study of robust self-testing (recall that this is the study of how the self-testing conclusions have to be modified when the observed correlations differ from the ideal ones).}, but it is indeed quite often possible: we are going to see an example with randomness amplification below.

\subsubsection{A very general tool: characterization of the quantum set}
\label{sssnpa}

At the beginning of Section \ref{secbell}, I reminded that the observed statistics ${\cal P}_{\cal X,\cal Y}$ can be obtained from measurement on a quantum state if there exist a state and suitable POVMs such that $P(a,b|x,y)=\textrm{Tr}(\rho\,E_a^x\otimes E_b^y)$. In this case, one says that ${\cal P}_{\cal X,\cal Y}$ \textit{belongs to the quantum set} $\cal Q$ for the scenario $({\cal X},{\cal A};{\cal Y},{\cal B})$. The dimension of the Hilbert spaces is left arbitrary. Therefore, assessing whether ${\cal P}_{\cal X,\cal Y}\in{\cal Q}$ is a possible device-independent test. As such, this test won't be often performed on observed data, insofar as nobody expects to observe statistics in flagrant violation of quantum physics. However, \textit{the characterization of the quantum set is a crucial tool} for more relevant device-independent tests: all those, in which one aims at optimizing some value over all possible quantum realizations.

But how to characterize the quantum set? The question looks similar to the characterization of LV models, but the mathematical answer is very different. This is due to the fact that the quantum set is convex but has continuously many extremal points\footnote{The proof of convexity is rather simple: since the dimension is not fixed, given two quantum points $P_1$ and $P_2$, one can always think of the underlying states and measurements to live in disjoint direct sums. Then, the point $pP_1+(1-p)P_2$ is obtained by measuring the obvious mixed state with the suitably extended measurements. The fact that there are continuously many extremal points is proved by exhibiting an explicit family of such points: I do not think this proof instructive enough as to be reported here.}. Its boundaries are not hyperplanes, like the facets of a polytope. The tool we are going to discuss allows to \textit{approach the boundaries of the quantum set from outside}.

Let us start with the following observation:
\begin{lemma}\label{lemmanpa}
Let $\{F_1,..., F_n\}$ be a collection of operators. The orthogonal matrix $M$ whose entries are
\ba
[M]_{ij}&=&\Tr(\rho F_i^\dagger F_j)
\ea is non-negative.
\end{lemma}
\begin{proof}
The proof is very direct: for any vector $\vec{v}\in\real^n$ it holds
\ban
\vec{v}^{\,T}M\vec{v}&=&\Tr\Big[\rho \big(\sum_iv_iF_i^\dagger\big)\big(\sum_jv_j F_j\big)\Big]\,\geq 0
\ean because both $\rho$ and any operator of the form $C^\dagger C$ are positive.
\end{proof}
Therefore, $M\geq 0$ for all $M$ constructed as above is a \textit{necessary} condition for ${\cal P}_{\cal X,\cal Y}\in{\cal Q}$. The idea is therefore to construct matrices of this type, whose entries can be expressed in terms of the observed statistics ${\cal P}_{\cal X,\cal Y}$, and check if they are not negative.

At this point, it is useful to break the chain of general reasoning and consider the same example of paragraph \ref{ssscase}: the correlators in the CHSH scenario. The observed data are $(E_{00},E_{01},E_{10},E_{11})$. In the quantum formalism,
\ba
E_{xy}=\Tr\Big[\rho (\Pi_{a=+1}^{x}-\Pi_{a=-1}^{x})\otimes(\Pi_{b=+1}^{y}-\Pi_{b=-1}^{y})\Big]\,.
\ea
We want to construct $M$ that can accommodate the observed data. Let us define
\ban
F_1&=&\left[\Pi_{a=+1}^{x=0}-\Pi_{a=-1}^{x=0}\right]\otimes \one\\ F_2&=&\left[\Pi_{a=+1}^{x=1}-\Pi_{a=-1}^{x=1}\right]\otimes \one\\
F_3&=&\one\otimes\left[\Pi_{b=+1}^{y=0}-\Pi_{b=-1}^{y=0}\right]\\ F_4&=&\one\otimes\left[\Pi_{b=+1}^{y=1}-\Pi_{b=-1}^{y=1}\right]\,;\ean notice that $F_j^2=\one$ because $\Pi_a^x\Pi_{a'}^x=\delta_{a,a'}\Pi_a^x$ and the analog for Bob. Then we have
\ba
M_Q&=&\left(\begin{array}{cccc} 1 & \Tr(\rho F_1F_2) & \Tr(\rho F_1F_3) & \Tr(\rho F_1F_4)\\
& 1 & \Tr(\rho F_2F_3) & \Tr(\rho F_2F_4)\\
& & 1 & \Tr(\rho F_3F_4)\\ & & & 1\end{array}\right)
\ea (since the matrix is symmetric, for clarity of reading I fill only the upper half). The observed data fit in this matrix as
\ba
M&=&\left(\begin{array}{cccc} 1 & u_1 & E_{00} & E_{01}\\
& 1 & E_{10} & E_{11}\\
& & 1 & u_2\\ & & & 1\end{array}\right).
\ea Two entries are undetermined, since we can't have observed data for $\Tr(\rho F_1F_2)$ and $\Tr(\rho F_3F_4)$: these numbers would require one of the parties to perform two measurements in the same run, which is against the operational rules of the black boxes. Nevertheless, if $(E_{00},E_{01},E_{10},E_{11})\in {\cal Q}$, there must exist two real numbers $u_1,u_2\in [-1,+1]$ such that $M\geq 0$. The calculation of the conditions under which this happens is available in the literature and not particularly instructive, but the result is remarkable: one can find $u_1,u_2$ such that $M\geq 0$ if and only if
\ba
|A_{00}+A_{01}+A_{10}-A_{11}|\leq \pi
\ea with $A_{xy}=\mathrm{Arcsin}(E_{xy})$. This non-linear equation approximates the boundary of the quantum set for correlators and has at times been called a ``quantum Bell inequality". An interesting case study are the statistics $(\frac{1}{\sqrt{2}},\frac{1}{\sqrt{2}},\frac{1}{\sqrt{2}},-\frac{1}{\sqrt{2}})$. It is the only quantum point that reaches $S=2\sqrt{2}$ and it saturates this inequality as well. On the one hand, this example shows that even elementary tests may enforce decent boundaries. On the other hand, we know (from self-testing) that a quantum point reaching $S=2\sqrt{2}$ must also have the marginals of the singlet, namely $P(a|x)=P(b|y)=\demi$: it is then simple to construct valid ${\cal P}_{\cal X,\cal Y}$ with the same correlators but biased marginals, which are therefore not quantum but cannot be detected by this test.

Going back to the general discussion, we left the problem as unparametrizable as before: Lemma \ref{lemmanpa} does not constrain the family of operators for which one should check $M\geq 0$; and even if all operators are checked, we have nothing more than a necessary condition for ${\cal P}_{\cal X,\cal Y}\in{\cal Q}$. Fortunately, it has been proved by Navascu\'es, Pironio and Ac\'{\i}n (NPA) that there exist a \textit{convergent hierarchy of criteria} \cite{npa}. The simplest test in the hierarchy, which is however already tighter than our example above, takes as $F$'s the identity $\one$ and all the measurement operators $E_a^x\otimes\one\equiv\Pi_a^x$ and $\one\otimes E_b^y\equiv\Pi_b^y$. The set of the ${\cal P}_{\cal X,\cal Y}$ for which $M\geq 0$ in this step is denoted ${\cal Q}_1$. The further steps of the hierarchy add to the list of $F$'s all the products or two (e.g. $\Pi_a^x \Pi_{a'}^{x'}$, $\Pi_a^x \Pi_b^y$), three (e.g. $\Pi_a^x \Pi_b^y \Pi_{b'}^{y'}$), etc. measurement operators. The set of the ${\cal P}_{\cal X,\cal Y}$ for which $M\geq 0$ at each stage are denoted ${\cal Q}_2$, ${\cal Q}_3$ etc. Clearly ${\cal Q}_1\supseteq{\cal Q}_2\supseteq{\cal Q}_3\supseteq...$ and ${\cal Q}_n\supseteq{\cal Q}$ for all $n$. What is not trivial to prove is that this hierarchy is convergent, i.e. $\lim_{n\rightarrow\infty}{\cal Q}_n={\cal Q}$. As such, the hierarchy of tests provides a necessary \textit{and sufficient} condition for ${\cal P}_{\cal X,\cal Y}\in{\cal Q}$.

All the steps of the hierarchy involve writing down a matrix with undefined elements and asking whether one can fill the gaps in such a way that it becomes non-negative. Such a computational problem is an example of a very well-known class known as \textit{semi-definite programs}, on which the reader will find abundance of information upon searching. Let me just stress here that semi-definite programs are not only efficiently solvable: the solutions are also exact up to numerical precision, because the result can be upper- and lower-bounded simultaneously. Needless to say, efficient as the algorithms may be, the size of $M$ grows very fast, so it is impossible to check the hierarchy up to arbitrary order. In practice, one compares the result after the first few steps with a known quantum result guessed to be optimal, and they often coincide, like the Tsirelson bound for CHSH in the example above.

A final remark: the hierarchy is elegantly defined and widely used, but one should keep in mind that Lemma \ref{lemmanpa} allows for full freedom in constructing a set of tests. In particular, there is no need to go all the way from ${\cal Q}_{n}$ to ${\cal Q}_{n+1}$ in order to strenghten the constraint: for example, it is enough to add a single $\Pi_a^x \Pi_b^y$ to the $F$'s that define ${\cal Q}_{1}$ to obtain a test which is \textit{a priori} stronger than the latter.

\subsection{Randomness amplification}

\subsubsection{From science to devices}

In paragraph \ref{ssmessage} I explained that the violation of Bell inequalities proves the existence of intrinsic randomness (unless one opts for a deterministic explanation that requires superluminal signaling). This suggests that Bell inequalities may be useful to \textit{generate random numbers for practical applications}. This possibility must be argued further.

The key issue is measurement independence. In paragraph \ref{ssmessage}, I argued that the possibility of measurement independence cannot be denied without denying a great part of the scientific method. Of course, this does not mean that measurement independence must be believed blindfoldedly for any Bell experiment: one can legitimately try and check that \textit{that} source and \textit{those} choice of settings are really independent. This check cannot be done in a device-independent way, it always leaves an element of trust. Usually, what one does is to implement the choice of the settings by ``random number generators" built by the trusted parties themselves. At first sight, this seems to render the generation of randomness through a Bell test a pointless task: if one has already got a random number generator, why bother using it as seed of a Bell test to create other random numbers, instead of using directly? Even further, many devices acting as random number generators are already produced and routinely used. So, is there an additional benefit in generating randomness using a Bell test? The answer is positive: the Bell-based protocol produces \textit{more}, and above all \textit{better}, randomness.

Indeed, notice that the choice of the settings needs not be ``random", but just uncorrelated from the source while the Bell test is running. If there are enough reasons to trust this to be true, then any source of weak, or even pseudo, randomness will do\footnote{For instance, Alice may take her second favorite book in its third French edition and select the settings based on the sixth letter of each even line. This is a perfectly deterministic recipe, but Alice may have good reasons to trust that the provider of the boxes will not guess it.}. But the product is very different: if a Bell inequality is violated, the outcomes of the Bell test is \textit{guaranteed in a device-independent way} and \textit{private}, since nobody can have an exact copy of Alice's list. This is nothing but a rephrasing of what the violation of Bell means: the outcomes could not possibly pre-exist, so they are guaranteed to be random \textit{and} nobody else can possibly have a copy, otherwise the list would have been pre-existing after all. Furthermore, the settings for a run need to be kept confidential only until the outcomes of that run are produced.

\subsubsection{Randomness amplification using CHSH}

Let us consider the CHSH test: if $S_{obs}>2$, there is some randomness in ${\cal{P}}_{{\cal X},{\cal Y}}$. For simplicity, we focus on Alice's outcomes only, so we want to compute\footnote{Even if it is quite obvious, let me stress that the intrinsic randomness of the Bell test is \textit{not} characterized by the \textit{observed} $P(a|x)$ from ${\cal{P}}_{{\cal X},{\cal Y}}$: indeed, one can easily define LV distributions with $P(a|x)=\demi$, for instance white noise.}
\ba\label{optprob}
P^*(a|x)\,=\,\max P(a|x)&\textrm{subject to} & \textrm{CHSH}[{\cal{P}}_{{\cal X},{\cal Y}}]=S_{obs}\;\textrm{and}\;{\cal{P}}_{{\cal X},{\cal Y}}\in{\cal Q}\,.
\ea This is the most elementary example discussed in \cite{randomnature}. We expect $P^*(a|x)=1$ for $S_{obs}\leq 2$ and $P^*(a|x)=\demi$ for $S_{obs}=2\sqrt{2}$ because we know that this condition self-tests the singlet.

The first constraint is linear, therefore easily dealt with. As we know from paragraph \ref{sssnpa}, the other constraint can be approximated by semi-definite criteria. In a systematic approach, therefore, one would start by replacing ${\cal Q}$ with ${\cal Q}_1$, solve the relaxed optimization with semi-definite programming, then check if one finds a quantum state and measurement that reach the same result. If this is the case, one has the solution. If not, one can iterate down the hierarchy. This could be a very instructive exercise for the reader. Here, we can solve the optimization \eqref{optprob} analytically using the decomposition of CHSH presented in paragraph \ref{sssblocks}.

Let us start with the observation that, since we do not bound the dimension of the Hilbert space, the quantum state can be considered pure without loss of generality\footnote{Moreover, it is intuitive that a pure state will give the optimal solution: in a mixed state, one adds classical randomness on top of the quantum randomness.}. Let us write it in the basis which decomposes Alice's and Bob's operators as in paragraph \ref{sssblocks}: $\ket{\Psi}=\sum_{\alpha,\beta}\sqrt{p_{\alpha\beta}}\ket{\psi^{\alpha\beta}}$ where $\ket{\psi^{\alpha\beta}}$ is the suitable normalized two-qubit state. Therefore,
\ba
S_{obs}&=&\sum_{\alpha,\beta} p_{\alpha\beta} \sandwich{\psi^{\alpha\beta}}{\hat{S}^{\alpha\beta}}{\psi^{\alpha\beta}}\,.
\ea
Similarly, any $P(a|x)$ compatible with the constraints will be of the form
\ba
P(a|x)&=& \sum_{\alpha,\beta} p_{\alpha\beta} \sandwich{\psi^{\alpha\beta}}{{\Pi_a^x}^{\alpha}\otimes \one^\beta}{\psi^{\alpha\beta}}\,.
\ea 
The problem has a standard form: maximize $P=\sum_kp_kP_k$ under the constraint $S=\sum_kp_kS_k\equiv S_{obs}$. In order to find the solution, we can first find the maximization for the elementary problem $\max_{S_k\equiv s}P_k=f(s)$. In the very reasonable case that $f''(s)$ does not change sign, we have two possibilities: (i) if $f$ is convex, the solution to the main problem is simply $P=f(S_{obs})$ obtained by setting all the $S_k=S_{obs}$; (ii) if $f$ is concave, the solution to the main problem is given by the convex combination of boundary points (for our problem, it would be $P=p\demi+(1-p)1$ with $p$ defined by $S_{obs}=p2\sqrt{2}+(1-p)2$).

So, our next step is to find the maximal value of $\sandwich{\psi}{{\Pi_a^x}\otimes \one}{\psi}$ over the set of \textit{pure two-qubit states}, under the constraint that $CHSH=S$. From \eqref{maxviolqubitpure}, we know that a state with Schmidt decomposition $\ket{\psi}=\cos\theta\ket{00}+\sin\theta\ket{11}$ can reach $s=2\sqrt{1+\sin^22\theta}$. By writing down the projector for that state, we see that its most biased marginal is $P(+|\hat{z})=\demi(1+\cos 2\theta)$. States with smaller values of $\theta$ can also reach the same $s$ (for suboptimal measurements), but their most biased marginal is definitely lower; states with higher values of $\theta$ cannot reach $s$. Therefore, by solving explictly for $\theta$, we find after trivial algebra $f(s)=\demi(1+\sqrt{2-(s/2)^2})$. This function being convex, we have found the solution to the general problem:
\ba
P^*(a|x)&=&\demi\,\left[1+\sqrt{2-\left(\frac{S_{obs}}{2}\right)^2}\right]\,.\label{randomchsh}
\ea
It is worth while noticing that the bound can be achieved by the observed $P(a|x)$ under ideal conditions. Indeed, as discussed in paragraph \ref{sschshqubits}, the maximal violation of CHSH for the two-qubit state $\ket{\Psi(\theta)}$ can be achieved by choosing $\hat{z}$ as one of Alice's measurements. If one creates exactly that state and performs exactly the optimal measurements, the observed $P(a|x)$ will be related to $S_{obs}$ through \eqref{randomchsh}.

\subsection{Subtle is the device, and maybe malicious}

Device-independent assessment works because \textit{one cannot violate Bell inequalities by accident}. Of course, since a Bell test is a statistical test, a violation may happen as a fluctuation, but the probability of such an event can be quantified exactly and this is not the point I want to make\footnote{Recall that, in this text, I have always assumed that the statistical character of quantum measurement is dealt with in a proper way.}. Suppose rather that one of the procedures for the Bell test is implemented incorrectly: for instance, the synchronization procedure may fail, so that Alice and Bob compare results that correspond to different pairs. Or suppose that there is a failure in Bob's hardware, in such a way that the choice of the input (the position of the knob) does not change the measurement that is really done\footnote{As a side remark, notice how such a hardware failure could easily lead to observe a violation of the uncertainty relations. Indeed, if Bob believes that he is alternating measurements of $\hat{X}$ and $\hat{P}$, while in fact the same observable is being measured all the time, then the observed variances will be equal and can be arbitrarily small.}. In either case, no violation will be observed: the assessment will produce the conservative answer that the device does not reach up to standards.

In fact, \textit{the conclusion of an allegedly device-independent assessment can be thwarted only if (i) there are loopholes in the Bell test or (ii) there is a failure related to the task itself to be assessed}. These are the points that one has to check.

We know already how to deal with point (i), but it is useful to revisit those items in this new light: 
\begin{itemize}
\item The settings must be chosen by the users independently from the devices, so the users should trust the random number generators that make those choices (ideally by fabricating them themselves).
\item Cheating boxes could try and exploit the \textit{detection loophole}. By requesting that the boxes give an outcome each time that a setting is chosen, this cheat can be ruled out easily\footnote{The only nuisance, not a minor one at the moment of writing, is that very few experimental setups can exhibit a violation of Bell inequalities while closing the detection loophole. As a result, device-independent assessment is currently challenging, if not strictly unfeasible. But there is no reason to believe that technology won't improve in the coming years.}.
\item One should exclude \textit{signaling} between the measurement devices, as we discussed at the end of paragraph \ref{ssmessage}.
\item \textit{Deviation from the i.i.d. scenario} can be treated as the memory effects in paragraph \ref{ssgill} provided that the scenario is not adversarial, i.e. if device-independent assessment is made with the purpose of quantifying defects. If the scenario is adversarial, however, the tools are just being developed\footnote{For instance, see \cite{vv11} for the certification of Bell-based randomness against a quantum adversary.} and exceed the scope of this text.
\end{itemize}

Point (ii) arises because we are not aiming at describing the violation itself, but at accomplishing a task with uncharacterized devices. For instance, consider the amplification of randomness, and assume that the Bell test is done by proper measurement of quantum entanglement, without any communication. Nevertheless, if a measurement box contains a signaling device, it can just leak out the list of random numbers at the end of the process. The random list is still guaranteed and fresh, but no longer private. Clearly, spacelike separation of the choices during the Bell test does not protect against such classical leakage. It is crucial to stress, however, that this is \textit{not} a limitation of the device-independent assessment: whenever there is a claim of privacy, one must check at best, and ultimately trust, that the devices are not unduly leaking out information.

\section{The robustness of quantum knowledge}
\label{secrobust}

The LV model for quantum correlations is definitely falsified by the violation of Bell inequalities. Logically, this does not force one to accept the quantum description of statistics \eqref{qprob} as the only alternative: all that this means is that some $P(a,b|x,y,\lambda)$ in \eqref{myfavorite} must violate Bell inequalities. In particular, one may still want to try and describe the observed statistics \textit{while recovering at least some of the classical features} that quantum theory denies. In this section, I am going to review further results that show the robustness of quantum knowledge: \textit{sheer observations, which quantum theory reproduces with simple finite-dimensional calculations, are sufficient to falsify many apparently reasonable alternative models}.

\subsection{Two no-signaling relaxations of LV}

For definiteness, let us focus on the statistics predicted by quantum theory for measurement on a singlet state:
\ba
{\cal P}_{Q}&=&\left\{P(a,b|\vec{a},\vec{b})=\frac{1}{4}(1- ab\,\vec{a}\cdot\vec{b})\,,\,a,b\in\{-1,+1\}, \vec{a},\vec{b}\in \mathbb{S}^2\right\}\,.
\label{statssinglet}\ea
As we know by now, they can't be reproduced with LV; moreover, if one accepts the validity of quantum theory, a small finite subset of them is sufficient for self-testing. Here we are going to show two much more stringent results. They are theory-independent like Bell's original theorem; and device-independent \textit{a fortiori}. Before discussing them, I need to introduce another useful inequality, because the CHSH inequality, with its two settings and its Tsirelson bound, is not powerful enough to reach these conclusions.

\subsubsection{A tool: the chained inequality}

The \textit{chained inequality} is a bipartite inequality with $M_A=M_B=M$ settings and two outputs on each side, based on the following sum of $2M$ terms
\ba
C_M&=&P(a_1=b_1)+P(b_1=a_2)+P(a_2=b_2)+...\nonumber\\
&&...+P(a_M=b_M)+P(b_M\neq a_1)\label{chained1}
\ea
where $P(a_x,b_y)$ is a convenient shorthand for $P(a,b|x,y)$ and $P(a_j=b_k)=P(+1,+1|j,k)+P(-1,-1|j,k)$, $P(a_j\neq b_k)=P(+1,-1|j,k)+P(-1,+1|j,k)$. The assumption of LV enforces the bound $C_M\leq C_{M,L}=2M-1$, while the algebraic bound (achievable with no-signaling distributions) is obviously $C_{M,NS}=2M$. As a kind of exercise, let me mention two alternative ways of writing the same inequality, which may be useful:
\begin{itemize}
\item Using $E_{jk}=2P(a_j=b_k)-1=1-2P(a_j\neq b_k)$, one can rewrite everything with correlators:
\ba
C'_M&=&E_{11}+E_{21}+E_{22}+E_{32}+...+E_{MM}-E_{1M}
\ea with $C'_{M,L}=2(M-1)$ and $C'_{M,NS}=2M$. From this expression, it is manifest that the chained inequality for $M=2$ is equivalent to CHSH.
\item For paragraph \ref{sssleg}, it is useful to invert all the terms of the equation using $P(a_j=b_k)+P(a_j\neq b_k)=1$ and get
\ba
C''_M=2M-C_M&=&P(a_1\neq b_1)+P(b_1\neq a_2)+P(a_2\neq b_2)+...\nonumber\\
&&...+P(a_M\neq b_M)+P(b_M= a_1)\,.\label{chained2}
\ea The local constraint reads $C''_M\geq C''_{M,L}=1$, the algebraic bound is $C''_{M,NS}=0$.
\end{itemize}
The chained inequality has many properties that one may consider as sub-optimal: it needs $M^2$ settings but only uses $2M$ correlators\footnote{This is not at all a problem for the theorist; but for the experimentalist, it means that most of the time one is recording data that won't actually be used.}, only one of which expresses a condition that is incompatible with the others under LV. Also,  for any $M>2$, the LV bound does not define a facet of the local polytope\footnote{The proof is simple. Using \eqref{dns}, we know that the polytope is embedded in a space of dimension $D_{NS}(M)=M^2+2M$. Facets are hyperplanes of dimension $D_{NS}-1$, so there must be $D_{NS}$ linearly independent points lying on each facet. In particular, there must be at least $D_{NS}$ extremal points on each facet. However, it is easy to check that only $4M$ extremal points saturate the chained inequality. Indeed, take \eqref{chained1}: $C_M=2M-1$ can be reached either by $a_1=b_1=a_2=...=b_M$ [in which case $P(b_M\neq a_1)=0$ and all the other probabilities are 1] or by a point of the type $a_1=b_1=...=a_k\neq b_k=a_{k+1}=...=a_M$ [in which case $P(a_k=b_k)=0$ and all the other probabilities are 1; there are clearly $2M-1$ positions for the $\neq$ sign]. There are certainly no other points, so, given that $a_1$ can take two values, we have found that $4M$ extremal point saturate the chained inequality as claimed. Since $4M<D_{NS}(M)$ for $M>2$, the chained inequality cannot be a facet.}, which means that there exist tighter inequalities for each $M$.

The main interest of the family of chained inequalities is that \textit{the quantum violation comes arbitrarily close to the algebraic one in the limit $M\rightarrow\infty$}. To prove this claim, it is enough to produce an example (which turns out to be the maximal violation achievable with quantum physics for each fixed $M$): consider a two-qubit maximally entangled state $\ket{\Phi^+}$ and the settings chosen taken in the $x-z$ plane of the Bloch sphere as
\ba
\vec{a}_j\,=\,\cos\theta_{2j-1}\hat{z}+\sin\theta_{2j-1}\hat{x}&,& \vec{b}_j\,=\,\cos\theta_{2j}\hat{z} + \sin\theta_{2j}\hat{x}
\ea where $\theta_k=\frac{k\pi}{2M}$. With this choice, all the probabilities in \eqref{chained1} become equal to $\frac{1}{2}(1+\cos\frac{\pi}{2M})=\cos^2\frac{\pi}{4M}$ and consequently
\ba
C_{M,Q}&=&M\left(1+\cos\frac{\pi}{2M}\right)\,\approx\,2M-\frac{\pi^2}{8M}\,.
\ea
This is the property that we are going to use to demonstrate the following two results.

\subsubsection{The singlet statistics have zero local fraction}
\label{sssepr2}

The first idea \cite{epr2,bkp06} consists in writing every $P$ in ${\cal P}_{{\cal X},{\cal Y}}$ as the convex sum
\ba
P(a,b|x,y)&=&p\,\int d\lambda\rho(\lambda)P(a|x,\lambda)P(b|y,\lambda)\,+\,(1-p)P(a,b|x,y,\mu)\nonumber\\
&\equiv&p\,P_L(a,b|x,y)\,+\,(1-p)\,P_{NS}(a,b|x,y)
\label{epr2}
\ea where $p\in[0,1]$, $P_L$ can be achieved with LV and $P_{NS}\equiv \frac{P-p_LP_L}{1-p_L}$ is only requested to be a valid probability distribution (it is no-signaling by construction, since both $P$ and $P_L$ are). The \textit{local fraction} $p_L$ of ${\cal P}_{{\cal X},{\cal Y}}$ is defined as
\ba
p_L&=&\max_{\textrm{\eqref{epr2} holds}} p\,.
\ea The local fraction is a natural figure of merit in terms of simulations: if one wants to simulate the observed correlations $P$, the local part $P_L$ comes for free. Clearly ${\cal P}_{{\cal X},{\cal Y}}$ violates at least one Bell inequality if and only if $p_L<1$; moreover, the observed violation of a Bell inequality puts an upper bound on $p_L$. Indeed, let $I_L$, $I_{obs}$ and $I_{\textrm{alg}}$ be the local, observed, and algebraic bounds for the inequality under study. From \eqref{epr2} it follows immediately that $I_{obs}\leq pI_L+(1-p)I_{\textrm{alg}}$, whence
\ba
p_L&\leq&\frac{I_{\textrm{alg}}-I_{obs}}{I_{\textrm{alg}}-I_L}\,. 
\ea
If $I_{obs}>I_L$, the bound on $p_L$ is not trivial. In particular, if ${\cal P}_{{\cal X},{\cal Y}}$ leads to $I_{obs}=I_{\textrm{alg}}$, then $p_L=0$. This holds for the singlet statistics \eqref{statssinglet} and the chained inequality in the limit $M\rightarrow\infty$, so we have proved
\begin{theorem}
The singlet statistics \eqref{statssinglet}, and even the subset obtained by restricting $\vec{a}$ and $\vec{b}$ to lie in a plane, have zero local fraction.
\end{theorem}

\subsubsection{The singlet statistics force fully random marginals}
\label{sssleg}

The second result is the achievement of a series of works initiated by Leggett. Let us recall the result of self-testing: if one assumes quantum physics to be valid, some observed ${\cal P}_{{\cal X},{\cal Y}}$ can only be due to the measurement of a maximally entangled state. Within quantum theory, this implies that the properties of the composite system are sharply defined while those of each sub-system are completely undefined. One may wonder whether a more classical picture can be recovered at the level of probabilities and try to reconstruct a given ${\cal P}_{{\cal X},{\cal Y}}$ using ${\cal P}_{{\cal X},{\cal Y},\lambda}$ that have \textit{biased marginals}. As an example, Leggett studied whether one can find a decomposition \eqref{myfavorite} of the singlet statistics \eqref{statssinglet}, whose marginals would look like those of pure single-qubit states $P(a|\vec{a},\lambda)=\demi(1+\vec{u}_\lambda\cdot\vec{a})$ and $P(b|\vec{b},\lambda)=\demi(1+\vec{v}_\lambda\cdot\vec{b})$. Generically, if the ${\cal P}_{{\cal X},{\cal Y},\lambda}$ are allowed to be signaling, they can even be completely deterministic. However, if the ${\cal P}_{{\cal X},{\cal Y},\lambda}$ are further restricted to be no-signaling (as in Leggett's example), very strong constraints can be proved, notably
\begin{theorem}
Consider any decomposition \eqref{myfavorite} of the singlet statistics \eqref{statssinglet}: if the $P(a,b|\vec{a},\vec{b},\lambda)$ are requested to satisfy the no-signaling constraint, then they must have fully random marginals, i.e. $P(a|\vec{a},\lambda)=P(b|\vec{b},\lambda)=\demi$ \cite{cr08}.
\end{theorem}

\begin{proof}
The proof uses the \textit{statistical distance} between two conditional probability distributions
\ba
D(P_{U|\omega},P_{V|\omega})&=&\demi\,\sum_{u\in U,v\in V}\big|P(u|\omega)-P(v|\omega)\big|\,.\label{distance}
\ea This distance is obviously symmetric by exchange of $U$ and $V$; it satisfies the triangle inequality as well as $D(P_{U|\omega},P_{V|\omega})\leq P(u\neq v|\omega)$.

Let us start from the chained inequality in its form \eqref{chained2} applied to one of the ${\cal P}_{{\cal X},{\cal Y},\lambda}\equiv {\cal P}_\lambda$. Using the bound just mentioned, we find
\ban
C''_M({\cal P}_\lambda)&\geq & D(P_{{\cal A}|11\lambda},P_{{\cal B}|11\lambda})+ D(P_{{\cal B}|21\lambda},P_{{\cal A}|21\lambda}) + D(P_{{\cal A}|22\lambda},P_{{\cal B}|22\lambda})+...\\
&&...+ D(P_{{\cal A}|MM\lambda},P_{{\cal B}|MM\lambda})+D(P_{{\cal B}|1M\lambda},1-P_{{\cal A}|1M\lambda})
\ean with obvious notations. Now we use the no-signaling property $P_{{\cal A}|xy\lambda}=P_{{\cal A}|x\lambda}$ and $P_{{\cal B}|xy\lambda}=P_{{\cal B}|y\lambda}$. Now we can apply the triangle inequality as
\ban
D(P_{{\cal A}|x=1\lambda},P_{{\cal B}|y=1\lambda})+ D(P_{{\cal B}|y=1\lambda},P_{{\cal A}|x=2\lambda})&\geq & D(P_{{\cal A}|x=1\lambda},P_{{\cal A}|x=2\lambda})
\ean and by repeated application we finally reach
\ba
C''_M({\cal P}_\lambda)&\geq &D(P_{{\cal A}|x=1\lambda},1-P_{{\cal A}|x=1\lambda})\nonumber\\&\stackrel{\eqref{distance}}{=}&\big|P(a=0|x=1,\lambda)-\demi\big|+\big|P(a=1|x=1,\lambda)-\demi\big|\,.
\label{uniform}\ea Now, the chained inequality is based on a linear expression, therefore \eqref{myfavorite} implies $C''_M({\cal P})=\int d\lambda \rho(\lambda) C''_M({\cal P}_\lambda)$. If ${\cal P}$ are the statistics \eqref{statssinglet} of the singlet, in theory we can have $C''_\infty({\cal P})= 0$. Since this value is the no-signaling bound, it implies $C''_\infty({\cal P}_\lambda)= 0$ for all $\lambda$, which, inserted in \eqref{uniform}, proves the theorem for $P(a|x=1,\lambda)$. The proof for all other settings is exactly the same: one just has to keep the suitable terms when iterating the triangle inequality.
\end{proof}

\subsection{Signaling models: simulation with communication}
\label{sssignal}

In the whole text, I adopted the view that the violation of Bell inequalities demonstrates intrinsic randomness. It matches how most physicists understand quantum physics and, as we have seen, it can even be related to potentially useful applications. This approach, however, does not address the reason why many (rightly or wrongly) feel uneasy with quantum physics: \textit{it tells you what you get out of it, but not how nature does it\footnote{After writing this text out of my head, I was reminded that Science Magazine published a contribution by Gisin with the title ``How does nature perform the trick?" \cite{gisinbellprize}. Probably the expression was in my subconscious.}}.

The violation of Bell inequalities leaves very little choice: any explanation of quantum correlations in classical terms must involve \textit{communication}. Moreover, the predictions of quantum theory do not vary if the choice of settings and the detection events are spacelike separated, and several experiments have confirmed the violation of Bell inequalities in this configuration: therefore, the hypothetical communication would have to be \textit{superluminal}. If one wants to tell ``how nature does it", the explanation must involve this rather problematic feature\footnote{Or better, recalling what we said in \ref{ssmessage}: the explanation must \textit{involve something that manifests itself as superluminal communication} in our common (3+1)-dimensional space-time.}. A comprehensive discussion would need to address the modifications of special relativity, a much debated topic that would bring us too far. Here, I address signaling models under the pragmatic angle of \textit{simulation}: never mind how nature really does it, which resources would \textit{we} need to simulate quantum statistics?

\subsubsection{How much?}
\label{sshowmuch}

Let us first study the \textit{amount of communication} required to simulate quantum statistics. One may naively guess this amount must be infinite, based on steering: in quantum theory, by choosing her measurement on an entangled state, Alice can prepare Bob's system in any state, and there are continuously many states. However, this reasoning fails because the simulators Anthony and Beatrix\footnote{Recall that I use the names of Alice and Bob for the users that choose measurement settings and observe outcomes. Think of Anthony and Beatrix as the mechanisms inside the boxes of Alice and Bob respectively.} are allowed to share some LV. In fact, at the moment of writing, no example of a quantum process is known, whose simulation would provably require an unbounded amount of communication (though some are conjectured). Here I present only the most famous result, due to Toner and Bacon \cite{tonerbacon}:
\begin{theorem} \label{tonerbaconthm}
The statistics \eqref{statssinglet} of the singlet state under all possible von Neumann measurements can be simulated with local variables and one bit of communication.
\end{theorem}

\begin{proof}
The proof is constructive and I present it in the version of \cite{degorre}. In each run, Anthony and Beatrix share two unit vectors $\vec{\lambda}_0$ and $\vec{\lambda}_1$, previously drawn with uniform distribution on the unit sphere $\mathbb{S}^2$. Anthony selects one of the two vectors according to the following rule: if $|\vec{a}\cdot\vec{\lambda}_0|\geq |\vec{a}\cdot\vec{\lambda}_1|$, he sets $\vec{\lambda}=\vec{\lambda}_0$; otherwise, he sets $\vec{\lambda}=\vec{\lambda}_1$. He communicates this choice to Beatrix (this is the bit of communication). Finally, Anthony outputs $a(\vec{\lambda})=\textrm{sign}(\vec{a}\cdot\vec{\lambda})$, Beatrix outputs $b(\vec{\lambda})=-\textrm{sign}(\vec{b}\cdot\vec{\lambda})$.

Clearly $\moy{a}=\moy{b}=0$, so now we need to prove that
\ba
\moy{ab}\,=\,\int_{\mathbb{S}^2} d\vec{\lambda}\rho(\vec{\lambda})\,a(\vec{\lambda})b(\vec{\lambda})&=&-\vec{a}\cdot\vec{b}\,.
\ea The important piece is the effective probability distribution $\rho(\vec{\lambda})$, which is not uniform, because Anthony's initial selection biases $\vec{\lambda}$ to be close to $\vec{a}$. Concretely, we shall prove in Lemma \ref{lemmadegorre} below that $\rho(\vec{\lambda})=\frac{1}{2\pi}|\vec{a}\cdot\vec{\lambda}|$. So we have
\ban
\moy{ab}\,=-\,\int_{\mathbb{S}^2} d\vec{\lambda}\,\frac{1}{2\pi}\,\underbrace{|\vec{a}\cdot\vec{\lambda}|\textrm{sign}(\vec{a}\cdot\vec{\lambda})}_{=\vec{a}\cdot\vec{\lambda}}\,\textrm{sign}(\vec{b}\cdot\vec{\lambda})
\ean and the result is readily derived by passing in spherical coordinates chosen such that $\vec{a}\equiv \vec{z}$ and $\vec{b}\equiv \cos\beta\vec{z}+\sin\beta\vec{x}$. \end{proof}

We have postponed the proof of the following
\begin{lemma}\label{lemmadegorre}
The selection procedure used in the proof above leads to an effective distribution $\rho(\vec{\lambda})=\frac{1}{2\pi}|\vec{a}\cdot\vec{\lambda}|$.
\end{lemma}

\begin{proof}
Consider first the following procedure, known as rejection method:
\begin{enumerate}
\item Pick $\vec{\lambda}_0$ uniformly on $\mathbb{S}^2$ and $u_0$ uniformly in $[0,1]$;
\item Keep $\vec{\lambda}_0\equiv\vec{\lambda}$ if $|\vec{a}\cdot\vec{\lambda}_0|\geq u_0$, discard it otherwise.
\end{enumerate}
The probability that a given $\vec{\lambda}_0$ is kept is the probability that $u_0$ smaller than $|\vec{a}\cdot\vec{\lambda}_0|$ is drawn; since $u_0$ is drawn uniformly, this probability is just $|\vec{a}\cdot\vec{\lambda}_0|$. Therefore $\rho(\vec{\lambda})\propto|\vec{a}\cdot\vec{\lambda}|$. By normalizing \textit{a posteriori}, one finds the factor $\frac{1}{2\pi}$.

The problem with this procedure, as the name indicates, is that several $\vec{\lambda}_0$ are discarded (in the context of this paper, it would amount to a detection loophole scheme in which Alice's box refuses to reply if her local variable does not match a desired condition). The introduction of $\vec{\lambda}_1$ solves the problem. Indeed, if $\vec{\lambda}_1$ is chosen uniformly in $\mathbb{S}^2$, $u_0\equiv |\vec{a}\cdot\vec{\lambda}_1|$ is uniform in $[0,1]$. So the procedure with which Anthony selects $\vec{\lambda}_0$ is the rejection method. By symmetry, the procedure with which Anthony selects $\vec{\lambda}_1$ is also the rejection method. Therefore the final $\vec{\lambda}$ is distributed according to $\rho(\vec{\lambda})=\frac{1}{2\pi}|\vec{a}\cdot\vec{\lambda}|$ as claimed, while no instance is discarded.\end{proof}

Let me finish by a small balance. On the one hand, it is quite remarkable that the statistics of the singlet, that are so strongly non-classical according to several criteria presented in the previous sections, are only one bit away from being classical in terms of communication. On the other hand, this and all signaling models have an unpleasant taste of fine-tuning\footnote{I thank Rob Spekkens for the expression.}: indeed, the use of the bit of communication in the simulation above is entirely \textit{ad hoc} and justified only \textit{a posteriori} by the fact that it reproduces the quantum statistics. Worse, almost any deviation from that rule would manifest itself in signaling.

\subsubsection{How fast?}

In order to reproduce all the predictions of quantum theory, the superluminal communication should in fact have infinite speed, because the entangled systems can be arbitrarily far apart. Remarkably, one can give a device-independent proof of this constraint, at least within a reasonable scenario:
\begin{theorem}\label{thmbancal}
Consider a theory that simulates quantum statistics with superluminal communication in a preferred frame. Then, either the speed of communication is infinite, or the theory must predict the possibility of sending messages faster than light (``observable signaling") for some arrangement of measurements in spacetime. The conclusion can be based only on observed statistics and does not require the candidate preferred frame to be identified.
\end{theorem}

In order to understand the theorem and how one can possibly prove such a result, we have to start by describing the signaling model in some detail. The \textit{desiderata} are:
\begin{itemize} \item[(D1)] Quantum statistics are simulated by LV, supplemented by a superluminal communication propagating at speed $v<\infty$ in a preferred frame;
\item[(D2)] The theory does not predict the possibility for us to send a message faster than light (no observable signaling).
\end{itemize} 

Let ${\cal E}_{A}$ denote the event in spacetime at which Alice chooses her measurement and gets her outcome (for simplicity, I suppose that these procedures take negligible time). Consider now a bipartite experiment. If $v<\infty$, three arrangements are possible: ${\cal E}_{A}\stackrel{v}{\rightarrow} {\cal E}_B$, ${\cal E}_{B}\stackrel{v}{\rightarrow} {\cal E}_A$, or ${\cal E}_{A}$ and ${\cal E}_B$ are outside each other's $v$-cone. It is very reasonable to concretize (D1) as:
\begin{itemize}
\item[(D1a)] In the arrangements ${\cal E}_{A}\stackrel{v}{\rightarrow} {\cal E}_B$ or ${\cal E}_{B}\stackrel{v}{\rightarrow} {\cal E}_A$, the observed statistics are compatible with a quantum state, further assumed to be the same in both cases\footnote{The fact that the state is the same does not play any role in what follows, but it is what one expects from quantum physics, in which the time-ordering of the measurements does not change the state. This whole requirement is another instance of the fine-tuning of signaling models: obviously, once the Pandora box is opened and a signal is allowed, there is no compelling reason \textit{a priori} to impose the observation of quantum statistics at all; but \textit{a posteriori}, we know that quantum statistics are observed and we have set out precisely to try and simulate them.}.
\item[(D1b)] If ${\cal E}_{A}$ and ${\cal E}_B$ are outside each other's $v$-cone, the observed statistics must be reproducible with LV. Notice that this arrangement never occurs if $v=\infty$.
\end{itemize}
Such a theory predicts a departure from quantum theory in the second arrangement and can therefore be tested in principle. Indeed, Alice and Bob can first arrange ${\cal E}_{A}\stackrel{v}{\rightarrow} {\cal E}_B$ and check that a Bell inequality is violated (with additional knowledge of the degree of freedom, they may even do a full tomography of the state). Then, Bob can bring his measurement to lie outside the $v$-cone of ${\cal E}_{A}$: the observed statistics should cease violating Bell inequalities. The design of the experiment may not be trivial, since we do not know which is the preferred frame, but there does not seem to be anything \textit{a priori} inconsistent in such a theory; in particular, it is trivial to find examples for which (D2) is satisfied. However, matters change when the theory is extended to \textit{three (or more) systems}.

\begin{figure}[ht]
\begin{center}
\includegraphics[scale=0.40]{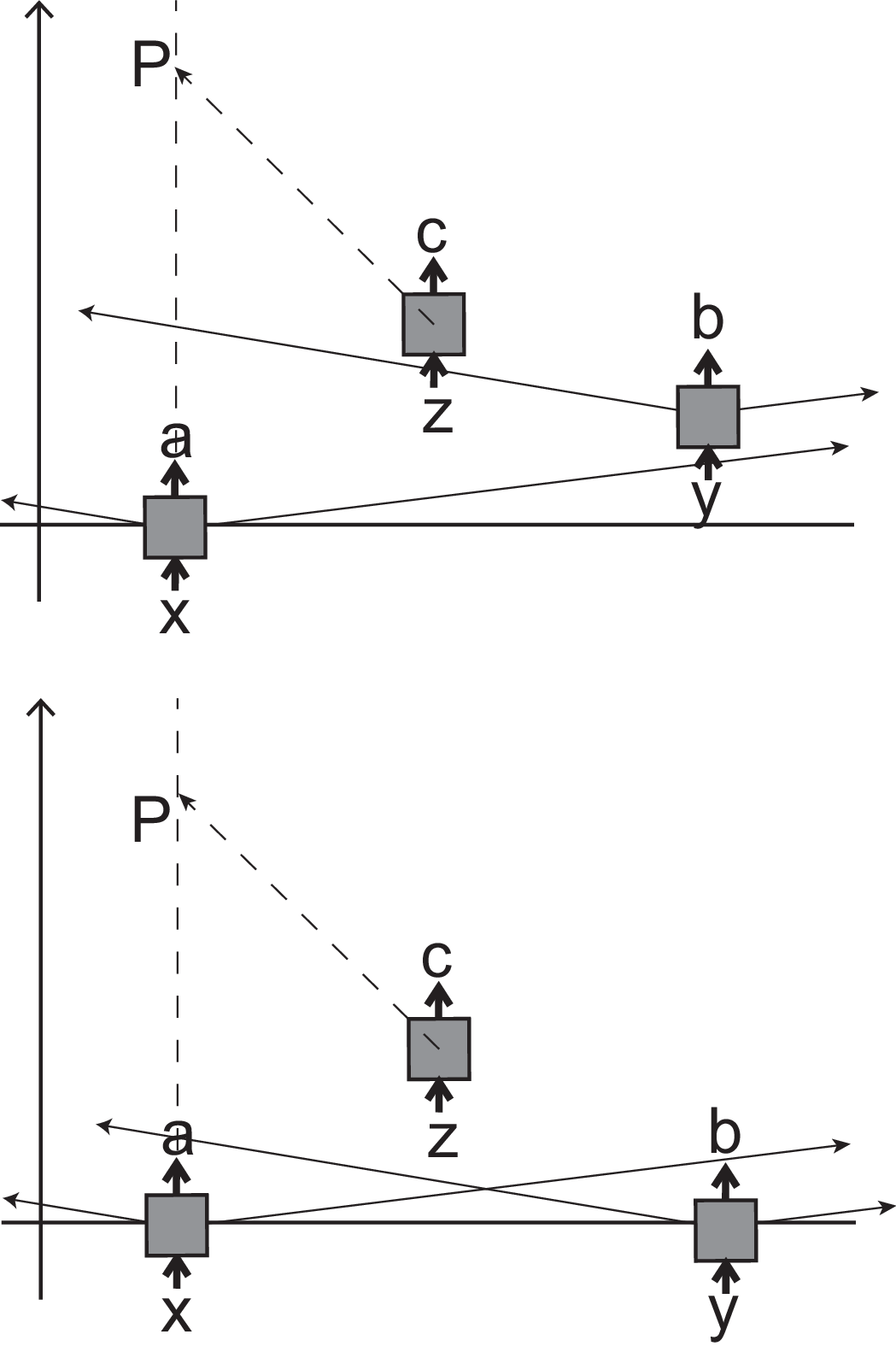}
\caption{The two possible arrangements described in the text (I top, II bottom). The thin full arrows represents the superluminal influences, the dashed arrow is a normal light signal. At point $P$, A can receive classical information about the setting and outcome of C, so correlations between A and C can be computed; while no light signal from B could have arrived to A. Therefore, the correlations A-C must be the same in both arrangements, otherwise faster-than-light communication would become possible from B to A (simply by B deciding to perform his measurement earlier or later).}
\label{figinfluences}
\end{center}
\end{figure}

Indeed, consider the following two arrangements of three measurements at different locations, with $C$ located between $A$ and $B$ (Fig.~\ref{figinfluences}):
\begin{itemize}
\item[(I)] ${\cal E}_{A}\stackrel{v}{\rightarrow} {\cal E}_B\stackrel{v}{\rightarrow} {\cal E}_C$.
\item[(II)] ${\cal E}_{A}\stackrel{v}{\rightarrow} {\cal E}_C$ and ${\cal E}_{B}\stackrel{v}{\rightarrow} {\cal E}_C$, but ${\cal E}_{A}$ and ${\cal E}_B$ are outside each other's $v$-cone; moreover, ${\cal E}_C$ is space-like separated from both ${\cal E}_{A}$ and ${\cal E}_B$ according to the usual light cones.
\end{itemize}
In configuration (I), one observes some quantum statistics ${\cal P}^I_{{\cal X},{\cal Y},{\cal Z}}$. In configuration (II), a departure from quantum statistics may happen if ${\cal P}^{I}_{{\cal X},{\cal Y}}$ violates Bell, because by construction ${\cal P}^{II}_{{\cal X},{\cal Y},{\cal Z}}$ must be such that
\ba
{\cal P}^{II}_{{\cal X},{\cal Y}}&& \textrm{is compatible with LV}\label{signallv}\,.
\ea
So far, we have exploited only (D1). Crucially, now (D2) imposes the \textit{two additional conditions}
\ba
{\cal P}^{II}_{{\cal X},{\cal Z}}\,=\, {\cal P}^{I}_{{\cal X},{\cal Z}} &\textrm{and} &
{\cal P}^{II}_{{\cal Y},{\cal Z}}\,=\,{\cal P}^{I}_{{\cal Y},{\cal Z}}\,,\label{signalm}
\ea for the simple reasoning explained in the caption of Fig.~\ref{figinfluences}. In turn, such conditions may impose constraints on the possible ${\cal P}^{II}_{{\cal X},{\cal Y}}$. Now one can hope to find a contradiction in the following way: start from some quantum statistics ${\cal P}^I_{{\cal X},{\cal Y},{\cal Z}}$ and compute ${\cal P}^{I}_{{\cal X},{\cal Z}}$ and ${\cal P}^{I}_{{\cal Y},{\cal Z}}$; if \textit{all} the statistics ${\cal P}_{{\cal X},{\cal Y},{\cal Z}}$ compatible with those marginals are such that ${\cal P}_{{\cal X},{\cal Y}}$ violates Bell, then no ${\cal P}^{II}$ can satisfy \eqref{signallv} and \eqref{signalm}. Therefore, one of the \textit{desiderata} must be dropped.

Such a three-partite example has not been found yet, but a \textit{four-partite example} exploiting similar contradictions has, thus proving Theorem \ref{thmbancal}. That example is rather complex: there is little added value in reproducing it here, compared to directing the reader to the published paper \cite{bancal2012}. I'd rather present here a partial proof \cite{braz05} that gives at least some intuition --- and, in the process, the reader can learn a nice quantum information result.

For this sake, I strengthen (D1) by requiring that any statistics of the theory, in particular ${\cal P}^{II}$, must be a quantum statistics. Also, I renounce full device-independence and assume that the systems are known, so tomography is possible. In this case, \eqref{signalm} is replaced by the stronger constraint
\ba
\rho^{II}_{AC}\,=\,\rho^{I}_{AC}&\textrm{and}& \rho^{II}_{BC}\,=\,\rho^{I}_{BC}\,.
\ea
The contradiction is then based on the following
\begin{lemma}
Consider a system $\compl^2\otimes\compl^2\otimes\compl^3$. For $0<\alpha<\frac{\pi}{2}$, there is only one quantum state such that
\ba
\rho_{AC}\,=\,\rho_{BC}&=& \demi\ket{\psi_1}\bra{\psi_1} \,+\,
\demi\ket{\psi_2}\bra{\psi_2}\label{rhos}\ea with
$\ket{\psi_1}=\sin\alpha\ket{00}+\cos\alpha\ket{12}$ and
$\ket{\psi_2}=\sin\alpha\ket{11}+\cos\alpha\ket{02}$: namely, the pure state $\ket{\Psi}=\cos\alpha\,\frac{\ket{01}+\ket{10}}{\sqrt{2}}\ket{2}
\,+\, \sin\alpha\,\frac{\ket{000}+\ket{111}}{\sqrt{2}}$. Moreover, this state is such that $\rho_{AB}$ violates the CHSH inequality for $\cos^2\alpha>\frac{1}{\sqrt{2}}$.
\end{lemma}

\begin{proof}

Since $\ket{\psi_1}$ and $\ket{\psi_2}$ are orthogonal, any purification of $\rho_{AC}$ can be written
\ban
\ket{\Phi}&=&\frac{1}{\sqrt{2}}\big(\ket{\psi_1}_{AC}\ket{E_1}_{BX}
+ \ket{\psi_2}_{AC}\ket{E_2}_{BX}\big) \ean with $X$ an auxiliary
mode and $\braket{E_1}{E_2}=0$. Further, since $B$ is a qubit, the Schmidt
decomposition yields \ban
\ket{E_1}_{BX}&=&c_0\ket{0}_B\ket{x_{10}}_X+c_1\ket{1}_B\ket{x_{11}}_X\\
\ket{E_2}_{BX}&=&d_0\ket{0}_B\ket{x_{20}}_X+d_1\ket{1}_B\ket{x_{21}}_X\ean
with $\braket{x_{k0}}{x_{k1}}=0$. We can insert these expressions into $\ket{\Phi}$ and compute the expression for $\rho_{BC}$, then require it to be given by (\ref{rhos}). Specifically, the requirement that $\rho_{BC}$ is orthogonal to $\ket{01}_{BC}$ and $\ket{10}_{BC}$ forces $c_1=d_0=0$, that in turn implies $c_0=d_1=1$. Using this condition, one further finds that $\rho_{BC}$ can be recovered if and only if $\braket{x_{10}}{x_{21}}=1$: therefore, $\ket{\Phi}_{ABCX}=\ket{\Psi}_{ABC}\ket{x}_X$, i.e. $\ket{\Psi}_{ABC}$ is the only quantum state, pure or mixed, compatible with the marginals (\ref{rhos}).

Now one can compute $\rho_{AB}$ and use \eqref{maxviolqubitgen} to prove that it violates the CHSH inequality if $\cos^2\alpha>\frac{1}{\sqrt{2}}$.
\end{proof}

Theorem \ref{thmbancal} proves that, in order to reproduce observed\footnote{Strictly speaking, the statistics used for the proof have not been observed yet. But they come from a set of few von Neumann measurements on a four-qubit state, and I don't see any reason to doubt the accuracy of quantum predictions in such a case.} statistics with communication, one has either to postulate a communication that propagates at infinite speed (in which case the universe is instantaneoulsy connected and everything is possible), or to conclude that faster-than-light signaling is possible after all (in which case not only quantum entanglement, but the whole of physics should be revisited). Thence this result comes as close as possible to a full falsification of signaling models.

\section{Towards a device-independent definition of quantum physics}
\label{secprinciples}

The device-independent outlook is fruitful to falsify alternative models and to promote Bell tests as certification tools. It is also encouraging to revisit old foundational questions, notably the one I present in this last section: what defines quantum physics?

\subsection{Traditional approaches to a definition of quantum physics}

The vast majority of presentations of quantum physics define it \textit{through its mathematical structure}, something that is often referred to as ``assuming the Hilbert space". Explicitly, the minimal assumption for the kinematics is that every physical property $P$ is described by a subspace ${\cal E}_P$ of a Hilbert space $\cal H$, with the rule that perfectly distinguishable properties are associated to orthogonal subspaces. Gleason's theorem\footnote{I am referring here to the original Gleason's theorem, and am extending its conclusions to the case $d=2$ which is not covered by the proof. There exist Gleason-like theorems covering the case $d=2$, and with much simpler proofs: the price to pay is that the assumptions (typically, the whole algebra of POVMs) become even harder to justify.} then leads to Born's probability rule, from which in turn follow all those other rules that are presented as ``axioms" in some introductory textbooks; the unitary representation of symmetries is another theorem, Wigner's. For the dynamics, the independent assumption of reversibility is required.

This approach is excellent for practical purposes, clarifying the mathematical structure that is accepted by everyone as a working tool. One can also accept this mathematical structure \textit{a posteriori}, based on its predictive power. Nevertheless, it is legitimate to ask if the Hilbert space assumption can be replaced by something more appealing \textit{a priori}. Even a simple review of the approaches that have been proposed would take us too long. But, in a way or another, they all \textit{assume the possibility of characterizing the state}: in other words, they assume that one can identify a closed set of measurements, such that the state is defined completely by the statistics of those measurements. As a consequence, before applying the formalism to a concrete case, one needs to have a pretty good idea of the degree of freedom under study and of the measurement devices that are in principle available.

All this is perfectly legitimate and common practice in physics. However, from the vantage point of device-independent assessment, we may hope to do even better. The hierarchy of semi-definite criteria described in paragraph \ref{sssnpa} goes only half way: it does define the possible physical observations in terms of a device-independent criterion (the positivity of some matrices built only on observed statistics); but there is no justification for this criterion, other than the fact that it recovers the statistics achievable with the Hilbert space formalism. Waiting for someone to find a physical reason, independent of quantum physics, why those matrices should be positive, I review here the partial successes achieved in answering the question: \textit{can quantum physics be defined in terms of device-independent physical principles?}

\subsection{No-signaling statistics}

\subsubsection{No-signaling as a framework}

As we have argued above, in order to simulate quantum statistics, we would need to use communication; but quantum physics achieves the same result without communication. It is tempting therefore to postulate ``no-signaling through observation" as a physical principle that must be respected in nature. This \textit{no-signaling principle} comes quite close to defining quantum physics itself, by cutting out all the models based on communication (which are unconstrained as for the statistics they can distribute). Popescu and Rohrlich \cite{PR} went further and asked whether the no-signaling principle defines quantum physics tightly. They found a counter-example (see next paragraph), so we need to find a more refined principle, or maybe a set of such principles. But all the following discussion will have the set of no-signaling statistics as underlying framework.

The set of no-signaling statistics is a polytope, since it is embedded in a finite-dimensional space and is defined by the linear constraints \eqref{defns}. Its extremal points are the local deterministic points (the same as for the local polytope) and some non-deterministic points. For the purpose of this text, we do not need to spend more time in these mathematical charcaterizations: I shall just present the simplest example of extremal no-signaling point that is not achievable with quantum statistics. It is the very example given by Popescu and Rohrlich and is nowadays generally called\footnote{When it comes to matters of priority, it had been mentioned in the literature as early as 1985 by other authors, see \cite{ourreview} for the details.} \textit{PR-box}.

\subsubsection{The PR-box}

We focus on the CHSH scenario of two parties, two inputs and two outputs. In paragraph \ref{ssscase}, we have noticed that the correlation vector $w=(+1,+1,+1,-1)$ would reach the algebraic maximum $S=4$ of CHSH; shortly later, we have established the Tsirelson bound $S=2\sqrt{2}$. This means that the correlations $w$, which can be compactly written as
\ba
a\oplus b=xy &\textrm{for}& a,b,x,y\in\{0,1\}\,,\label{prcorr}
\ea cannot be distributed with quantum physics.

There are four deterministic points that achieve $w$:
\ban
{\cal D}_1&:&P(a,b|0,0)=P(a,b|0,1)=P(a,b|1,0)=\delta_{a=0,b=0},\\&& P(a,b|1,1)= \delta_{a=0,b=1}\,;\\
{\cal D}_2&:&P(a,b|0,0)=P(a,b|0,1)=P(a,b|1,0)=\delta_{a=0,b=0},\\&& P(a,b|1,1)= \delta_{a=1,b=0}\,;\\
{\cal D}_3&:&P(a,b|0,0)=P(a,b|0,1)=P(a,b|1,0)=\delta_{a=1,b=1},\\&& P(a,b|1,1)= \delta_{a=0,b=1}\,;\\
{\cal D}_4&:&P(a,b|0,0)=P(a,b|0,1)=P(a,b|1,0)=\delta_{a=1,b=1},\\&& P(a,b|1,1)= \delta_{a=1,b=0}\,.
\ean
But each of these deterministic points violates the no-signaling constraint. Consider for instance ${\cal D}_1$: Bob can directly read Alice's input by choosing $y=1$, since $b_{y=1}=x$. It is easy to check that only one convex combination of these points satisfies the no-signaling condition \eqref{defns}, namely the PR-box
\ba
{\cal P}_{PR}&:&P(a,b|0,0)=P(a,b|0,1)=P(a,b|1,0)=\demi\delta_{a=0,b=0}+\demi\delta_{a=1,b=1},\\&& P(a,b|1,1)= \demi\delta_{a=0,b=1}+\demi\delta_{a=1,b=0}\,.
\ea This uniqueness, together with the fact that $S=4$ is the maximum CHSH can reach, immediately implies that ${\cal P}_{PR}$ is an extremal point of the no-signaling polytope.

In the last decade or so, intriguing results have been obtained by considering the PR-box and its generalization as \textit{resources} for distributing correlations: the interested reader will find basic information and references in \cite{ourreview,fff}. In the following, I shall use the PR-box simply as the prototypical example of no-signaling statistics that cannot be achieved with quantum physics.

\subsection{Device-independent physical principles}

\subsubsection{Information causality}

Information causality (IC) may be the physical principle that defines quantum physics, but we have not been able to prove it yet; for sure, it is the attempt that comes closer to selecting the quantum set within the no-signaling polytope. Here, I limit myself to presenting the initial intuition, because Marcin Pawlowski (the one who had first the idea) and I have recently written a synthetic text \cite{ps13} which is, in my opinion, as clear as it gets\footnote{I am still able to use the copy-and-paste function of my computer, but I don't see any point in using it here.}.

\begin{figure}[ht!]
\begin{center}
\includegraphics[scale=0.50]{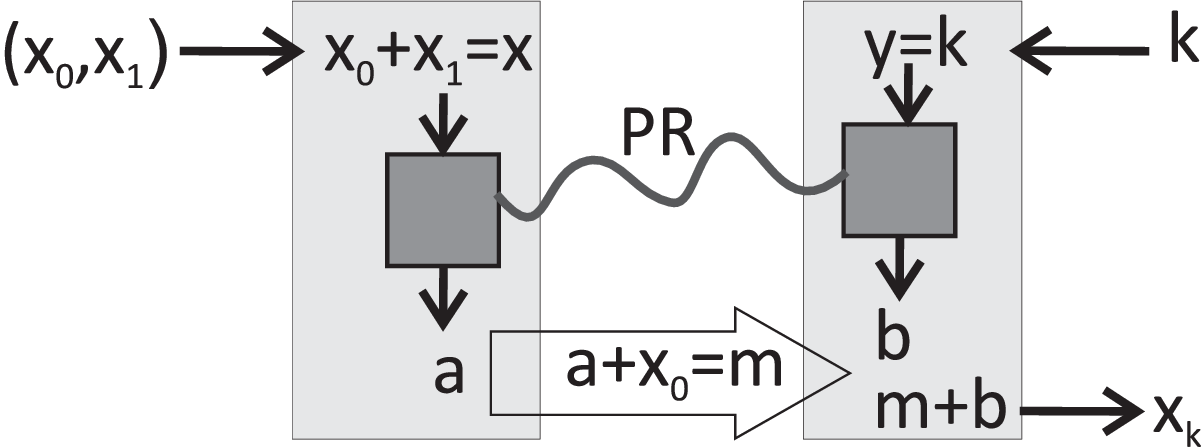}
\caption{The task that inspires information causality and how the PR-box achieves it.}
\label{figrac}
\end{center}
\end{figure}

It all starts by finding something that would ``go wrong" if PR-boxes would exist. The starting point to formulate IC is the power of the PR-box in the communication task\footnote{This taks is known in information science as the simplest example of ``random access code"; or, if one adds the requirement that Alice is forbidden to know Bob's choice even a posteriori, as ``oblivious transfer".} described in Fig.~\ref{figrac}. Alice's input consists of a pair of bits $(x_0,x_1)\in\{0,1\}^2$ drawn uniformly at random among the four possible values; Bob's input is a bit $y\in\{0,1\}$ unknown to Alice. The goal of the channel is to give Bob $x_y$ without giving him any knowledge of $x_{1-y}$.

Let us pause to examine this situation classically. If Alice cannot send any information to Bob, obviously Bob cannot retrieve either bit. If Alice can send two bits, Bob can trivially retrieve both $x_y$ and $x_{1-y}$. The interesting case is when Alice is \textit{restricted to send only one bit}: may they succeed in the task, possibly with the help of pre-established LV? One can prove that they can't: the best Alice can do with one bit of communication is the obvious strategy: she encodes by default $x_0$, sends it, and Bob outputs it. If Bob had received $y=0$, his output is correct; if he received $y=1$, his output is uncorrelated with the right answer $x_1$.

Remarkably though, Alice and Bob would succeed if, instead of pre-sharing LV, they would be allowed to pre-share a PR-box. Indeed, here goes the protocol: Alice inputs $x=x_0\oplus x_1$ in the PR-box; she gets the outcome $a$ and sends to Bob the single-bit message $m=a\oplus x_0$. Bob inputs $y$ in the PR-box, gets the outcome $b$ and produces the guess $\beta=m\oplus b=(a\oplus b)\oplus x_0$. By the rule \eqref{prcorr} of the PR-box, $a\oplus b=(x_0\oplus x_1)y$; so $\beta=x_0\oplus[(x_0\oplus x_1)y]=x_y$ as claimed.

Nothing has gone blatantly wrong: Alice sends one bit, Bob gets one bit, because the PR-box is a no-signaling resource. It is nevertheless puzzling that Bob can guess perfectly either of the two bits: it looks as if both bits had been transferred to his location, even if he is allowed to read only one. \textit{Information causality erects as a principle that this should not happen: positively, if Alice sends one bit, the amount of useful information at Bob's location cannot exceed one bit} \cite{icnature}. When formulated mathematically, it turns out that IC is respected if Alice and Bob would share entanglement. For the CHSH scenario, IC is violated by any no-signaling resource such that $S>2\sqrt{2}$: in other words, the Tsirelson bound is recovered without any reference to the algebra of Hilbert spaces.

\subsubsection{Macroscopic locality}

Like information causality, macroscopic locality \cite{nw09} is defined in terms of a restricted task. The scenario is that of a normal Bell experiment, with a source producing i.i.d. signals that would generate the ``microscopic" statistics ${\cal P}_{{\cal X},{\cal Y}}$. However, the boxes are not capable of measuring individual events and reconstruct those statistics: for each choice of their measurement settings, they are restricted to observe only ``macroscopic" averages.

Specifically, for any $x$, Alice can access only the currents $\mathbf{I}_x=(I_{a=0|x},I_{a=1|x},...)$ created in her $m_A$ detectors by sending $N$ signals. Similarly, Bob can access $\mathbf{J}_y=(J_{b=0|y},J_{b=1|y},...)$ defined in an identical way. This defines a scenario with the same number of settings as the microscopic one but a much larger alphabet of outcomes\footnote{For fixed finite $N$, $\mathbf{I}_x$ takes values in $\nat^{m_A}$ restricted by $\sum_{a}I_{a|x}=N$; and similarly for $\mathbf{J}_y$. But since we are going to consider the limit $N\rightarrow\infty$, the discrete and finite structure won't play any role.}. After repeating such an experiment many times, Alice and Bob can reconstruct the statistics of these macroscopic currents
\ba
{\cal{P}}^M_{{\cal X},{\cal Y}}&=&\left\{ P(\mathbf{I}_x,\mathbf{J}_y)\,,\;x\in{\cal X}\,,\,y\in{\cal Y}\right\}\,.
\ea
The \textit{principle of macroscopic locality} states that ${\cal{P}}^M_{{\cal X},{\cal Y}}$ should not violate any Bell inequality in the limit $N\rightarrow\infty$. The rationale is that \textit{coarse-graining should lead to classical physics}, an intuition to which many would unhesitatingly subscribe\footnote{A point of logic: in order to accept the principle, it is enough to consider coarse-graining at the measurement device as a \textit{possible} path for the emergence of the classical world. It does not need to be believed as \textit{the only} mechanism.}.

It is easy to show that the PR box violates macroscopic locality. The macroscopic statistics are given by
\ba
{\cal{P}}^M_{{\cal X},{\cal Y}}\,[\textrm{PR box}]&:&\left\{ \begin{array}{lcl}
(x,y)=(0,0)&:& \mathbf{I}_0=\mathbf{J}_0\\
(x,y)=(0,1)&:& \mathbf{I}_0=\mathbf{J}_1\\
(x,y)=(1,0)&:& \mathbf{I}_1=\mathbf{J}_0\\
(x,y)=(1,1)&:& \mathbf{I}_1=(I,N-I)\,,\,\mathbf{J}_1=(N-I,I)
\end{array}\right.\,.
\ea In order to show that these statistics violate some Bell inequality, we can append the following local post-processing (majority vote): for each run of the macroscopic experiment, Alice's box outputs $\alpha=0$ if $I_{a=0|x}-I_{a=1|x}\geq 0$ and $\alpha=1$ otherwise; Bob's box outputs $\beta$ with the same rule. Then, the statistics $P(\alpha,\beta|x,y)$ thus obtained define again the PR box. Loosely speaking, this proves that the PR box gets out unscathed from the coarse-graining process and therefore grossly violates macroscopic locality.

The extent to which the principle of macroscopic locality approximates the quantum set $\cal Q$ is known exactly:
\begin{theorem}
${\cal{P}}^M_{{\cal X},{\cal Y}}$ can be reproduced by LV if and only if ${\cal P}_{{\cal X},{\cal Y}}$ belongs to ${\cal Q}_1$, the set of statistics that pass the first test of the hierarchy of semi-definite criteria. In particular, since ${\cal Q}\subset{\cal Q}_1$, all quantum statistics respect macroscopic locality, but there are non-quantum statistics that respect it as well.
\end{theorem}

\begin{proof}
As shown in paragraph \ref{ssdet}, ${\cal{P}}^M_{{\cal X},{\cal Y}}$ can be reproduced by LV if and only if each of the $P(\mathbf{I}_x,\mathbf{J}_y)$ can be computed as marginal of a joint probability distribution ${\mathbf{P}}(\mathbf{I}_1,\mathbf{I}_2,...,\mathbf{I}_{M_A};\mathbf{J}_1,\mathbf{J}_2,...,\mathbf{J}_{M_B})$.

Now, notice that
\ba
I_{a|x}=\sum_{n=1}^N \delta_{a_x(n)=a}&,& J_{b|y}=\sum_{n=1}^N \delta_{b_y(n)=b}
\ea are sums of i.i.d. random variables. Each $\mathbf{I}_x$ being a vector of $m_A$ numbers $I_{a|x}$ and each $\mathbf{J}_y$ being a vector of $m_B$ numbers $I_{b|y}$, they are in turn sums of i.i.d. random vectors. If ${\mathbf{P}}$ exists, the \textit{central limit theorem} states that the fluctuations of its variables around their average must obey, in the limit $N\rightarrow\infty$, a multivariate Gaussian distribution with zero average and covariance matrix $\Gamma\geq 0$. Inversely, if such a distribution exists, it defines a valid ${\mathbf{P}}$.

The proof will be finished by showing that $\Gamma$ is essentially the matrix $M$ that defines ${\cal Q}_1$ for ${\cal P}_{{\cal X},{\cal Y}}$, as defined in paragraph \ref{sssnpa}. Indeed, let's define the fluctuations as
\ba
f_{a|x}=\frac{I_{a|x}-\moy{I_{a|x}}}{\sqrt{N}}&,& f_{b|y}=\frac{J_{b|y}-\moy{J_{b|y}}}{\sqrt{N}}\,:
\ea the entries of the covariance matrix are the $\Gamma_{ij}=\moy{\mathbf{f}_{i}\cdot\mathbf{f}_{j}}$, i.e. with a suitable labeling
\ba
\Gamma &=&\left(\begin{array}{c|c} \left\{\moy{\mathbf{f}_{x}\cdot\mathbf{f}_{x'}}\right\} & \left\{\moy{\mathbf{f}_{x}\cdot\mathbf{f}_{y'}}\right\}\\\hline
\left\{\moy{\mathbf{f}_{y}\cdot\mathbf{f}_{x'}}\right\} & \left\{\moy{\mathbf{f}_{y}\cdot\mathbf{f}_{y'}}\right\}
\end{array}\right)\,.
\ea The elements of the off-diagonal blocks are terms of the form
\ba
\moy{f_{a|x}f_{b|y}}&=&\frac{1}{N}(\moy{I_{a|x}J_{b|y}}-\moy{I_{a|x}}\moy{J_{b|y}})\nonumber\\& \stackrel{N\rightarrow\infty}{\longrightarrow}& P(a,b|x,y)-P(a|x)P(b|y)\,.
\ea The terms of the diagonal blocks have a similar form, but can be associated to observable probabilities only when $x=x'$ ($y=y'$). This matrix defines indeed the step ${\cal Q}_1$ of the hierarchy, up to redefining the measurement operators as $F_{a|x}=\Pi_{a}^{x}-\moy{\Pi_{a}^{x}}\one_{d_A}$ and $F_{b|y}=\Pi_{b}^{y}-\moy{\Pi_{b}^{y}}\one_{d_B}$. \end{proof}

It is important to stress the crucial role of the i.i.d. assumption, which is required in the proof above in order to invoke the central limit theorem:
\begin{itemize}
\item If i.i.d. is not requested, quantum states that violate a Bell inequality for the coarse-grained detection can be found (although, for large $N$, they may be impossible to realize in practice). So, the theorem does not say that coarse-graining \textit{alone} washes all Bell violation away.
\item Under the \textit{promise} that their source is producing i.i.d. signals, Alice and Bob can infer ${\cal P}_{{\cal X},{\cal Y}}$ from the covariance matrix of their observed macroscopic statistics. Thus, they may infer that their source could be used to violate some Bell inequality, if fine-grained detection would be available.
\end{itemize}

\subsubsection{A temporary balance}

Let me finish by summarizing where we stand in the quest for a complete device-independent definition of quantum physics. 

One path towards may pass through finding physical interpretations for the steps ${\cal Q}_n$ of the NPA hierarchy (so far, we know only the physical meaning of ${\cal Q}_1$). However, this hierarchy is just the only family of tests for which convergence has been proved: nothing guarantees that all its steps must have a clear physical meaning. In particular, information causality does not correspond to any step of that hierarchy and, insofar as we know, a suitable generalization of it may already define the quantum set exactly.

\section{Conclusion}

Set in a science museum, in a not too distant future:\\
-- ``Mum, what are those two boxes?"\\
-- ``They show the violation of Bell inequalities. You remember when you went with dad to buy pieces for the new quantum computer and he quarreled with the vendor? The fellow was trying to sell boxes of low quality, but your dad knows these things: nobody can cheat him."\\
-- ``The small shiny boxes?"\\
-- ``Yes, nowadays they are very small. This one here is the old kind, the one we had when I was a student. I played a bit with such boxes. They changed the way I look at nature."\\
-- ``Wow! How does it work?"\\
-- ``Nobody really knows."

\section*{Acknowledgments}

In the last decade, I have had the chance to meet virtually everyone working on these topics: I don't even try to make a list, because the only predictable result is that I would forget to include many. I limit myself to acknowledging the people who have contributed directly to the present text, which started as part of the lecture notes of a CQT graduate module. Berge Englert was greatly supportive in the usual bureaucratic struggles involved in setting up a new module. A special thought goes to the first batch of students who took the module and suffered some of the deficiencies of that first presentation: Aarthi Sundaram, Benjamin Phuah, Rakhita Chandrasekara, Dai Jibo, Do Thi Xuan Hung, Le Phuc Thinh, Teh Run Yan, Rus Whang, Sambit Bikas Pal and Tanvirul Islam. Precious feedback was given by Nicolas Gisin, Jeysthur Ang, Lana Sheridan, Yang Tzyh Haur, Mark Wilde, Jonathan Olson, Corsin Pfister and Jed Kaniewski. 

This work is supported by the National Research Foundation and Ministry of Education, Singapore.

\begin{appendix}

\section{Reading EPR again}
\label{aepr}

In most texts devoted to Bell inequalities, including this one, the original EPR paper \cite{epr} is quoted only at the beginning and never mentioned again. In this appendix, I revisit the \textit{physical process} proposed by EPR in the light of our present understanding.

The EPR state is a bipartite state of two particles on a line, each characterized by the Hilbert space $L^2(\real)$. It is immediate to check that the operators $x_1-x_2$ and $p_1+p_2$ commute; so one can define the state that satisfies both\footnote{One can as well study the case where $x_1+x_2$ and $p_1-p_2$ are used to define the state.}
\ba
x_1-x_2=d&\textrm{and} & p_1+p_2=u\label{defepr}
\ea with $d,u\in\real$. Explicitly, this state is such that:
\begin{itemize}
\item If position is measured, particle 1 can be found \textit{anywhere} and particle 2 will be found at $x_2=x_1-d$. If $d=0$, the particles are found at the same location, but this value does not play any role in the argument.
\item If momentum is measured, particle 1 can be found with \textit{any value} and particle 2 will be found to have $p_2=u-p_1$. If $u=0$, the particles have opposite momentum, but again this does not play any role.
\end{itemize} 

The \textit{EPR reasoning} is the following: if I measure the position of particle 1 and find $x_1=x$, I know for sure what the result of a measurement of position of particle 2 would be, namely $x_2=x-d$. So I can just as well learn something else by measuring momentum of particle 2: upon finding $p_2=p$, I know that a measurement of momentum on particle 1 would have given $p_1=u-p$. Notice that this does not contradict the uncertainty relations\footnote{The predictions for the EPR state, like any state in quantum theory, obey the uncertainty relations: for instance, the distribution of $p_1=u-p_2$ is uniformly spread on $\real$, even post-selecting on the runs in which one finds a given value $x_1=x$. Indeed, $\Delta x\Delta p \geq \frac{\hbar}{2}$ means that one cannot prepare a source, such that the statistics of both position and momentum are sharply defined. This does not imply logically that both position and momentum cannot be sharply defined \textit{in each single run}. Admittedly, if position and momentum were well-defined in each run, it would be hard to understand why nature conspires to hide this from us over many runs; nevertheless, the existence of LV models for measurements on single degrees of freedom cannot be denied, however ``unnatural" one may find them.}.

Our understanding of LV statistics is a powerful tool to analyze this reasoning: in fact, in a sense, we have already analyzed it. Indeed, we can rephrase it with the singlet state of two qubits: if I measure $\sigma_z$ on Alice's qubit and find $+1$, I know that a measurement of $\sigma_z$ on Bob's qubit would have hold $-1$. So I can just as well measure $\sigma_x$ on Bob's qubit, etc. We have seen in paragraph \ref{powerlv} that the statistics of these measurements \textit{can} be reproduced with LV; and so can the statistics of a measurement on a single particle. Therefore, my assessment is that EPR drew the right conclusion, based on their limited evidence: the existence of pre-determined values for single particles cannot be excluded, and the EPR reasoning even enforces this as the most reasonable explanation of some other predictions of quantum theory! Only the violation of Bell inequalities can shatter this vision.

As it turns out, EPR were not on the simplest track to formulate the problem in the right way:
\begin{itemize}
\item The mental image of the experiment is complex: particle 1(2) is not ``the particle that will be measured in the localized region called Alice's (Bob's) lab". There must be another label to identify the subsystems, for instance the type (particle 1 is a proton, particle 2 an electron), or the state of an internal degree of freedom (particle 1 is spin up, particle 2 is spin down). Here we can appreciate the contribution of Bohm, who made it possible to envisage observations in space-like separated labs by describing entanglement in internal degrees of freedom rather than in position and momentum.
\item Even if EPR had been aware of the inequalities and had tried to violate them, they were handling one of the most difficult states for the task. Indeed, the EPR state has a positive Wigner function: so, as long as both particles undergo arbitrarily many measurements of the type $\cos\theta x+\sin\theta p$, there exist a LV model that describes all the statistics. In order to violate Bell inequalities, the theorist may rewrite the EPR state in the basis of eigenstates of the harmonic oscillator, thus going back to the discrete-variable formalism in which we proved Gisin's theorem.
\end{itemize}

Ultimately, I think that we must admit the evidence: EPR had a great intuition and pinpointed a crucial feature of quantum physics, but it took the work of Bell to put this intuition in its suitable conceptual framework.

\section{The tortuous path to device-independence}
\label{sspathdi}

I have presented the notion of device-independent assessment as a \textit{natural} consequence of the meaning of Bell inequalities --- and such, I am convinced, it is. However, as it often happens, we human beings need some time to straighten our thoughts. Here I am going to describe how we reached there.

\subsection{Prehistory}

No scientific idea is born out of nothing: there are hints, precursors, that someone may want to call ``lost opportunities" but really just prove that ground must be broken before a seed can become a tree.

The fact that $S=2\sqrt{2}$ identifies the singlet up to local isometries was proved as early as 1992 \cite{pr92,bmr92}, as the answer to the question ``which quantum states violate CHSH maximally". Nobody thought of turning it around and using CHSH to certify the presence of those states. Curiously, when someone had that idea, they did not use Bell inequalities: Mayers and Yao proposed their own scheme in 1998 and invented the term ``self-testing" for the occasion \cite{my04}. The Mayers-Yao papers contained all it gets to ignite the revolution, including the motivation by quantum cryptography; but failed to. I can only try and guess why: quantum cryptography, it was still in its very early days. Not many could follow Mayers' very complicated proof of ``unconditional security" that assumed qubits and well-defined commutation relations. It is safe to guess that even fewer people would embark in studying an equivalently complicated proof of self-testing apparatuses!

The Vienna-Gdansk collaboration were among those who, like me, were still clinging to the hope that Bell was not a piece for a museum. In 2004 they published a paper relating Bell inequalities to communication complexity games \cite{vienna}. With hindsight, the approach is a bit artificial: the games were defined in terms of Bell inequalities, nobody would have invented them otherwise. Still, I can't refrain form quoting the last sentence of the abstract: ``Thus, violation of Bell's inequalities has a significance beyond that of a non-optimal-witness of non-separability".

As for myself, the first contact with device-independence dates from a visit that Nicolas Gisin and I were paying to Sandu Popescu in Bristol, probably in 2002. Nicolas, who was starting to go commercial with quantum cryptography, was concerned about certification. We were having a moderately calm discussion about this topic, Sandu being notoriously not excited by those practical developments and participating more out of friendship. The conclusion was simple: the ultimate certification must come from Bell inequalities. There and then, this conclusion, obvious for anyone who has understood Bell, did not seem particularly deep to us. We had simply in mind a two-step procedure, in which one would first test the quantumness of the device using Bell, then run a usual protocol: not an idea to write a single paper about. Why did we not think further? Nicolas and I just don't know, Sandu has probably forgotten the incident.

\subsection{Making history}

The turning point seems to be the year 2004. Quantum cryptography had grown to a quite mature field, with several experimental groups improving their technologies and theorists harnessing the techniques of security proofs. Also, some more abstractly-minded theorists had started playing with PR-boxes and their generalizations. In this context, Barrett, Hardy and Kent had a remarkable insight: maybe one can prove security based only on no-signaling. This would mean that ``quantum" cryptography would actually remain secure in all post-quantum theories that do not allow signaling by mere observation. They found the first example of such protocol \cite{bhk}.

Nicolas Gisin, who was my boss at that time, took their result very seriously and started working on it with Toni Ac\'{\i}n and Lluis Masanes from Barcelona. Among other things, they noticed something that the whole quantum information community had got wrong for years. By common knowledge, the BB84 and the Ekert protocols for quantum cryptography were considered equivalent, one being just an entanglement-based version of the other. But this is true only under the assumption that one is using two-qubit states and mutually unbiased measurements. In reality, the Ekert protocol can provide device-independent security, because it is based on CHSH; while the security of BB84 and its entanglement-based version BBM is fully compromised in a device-independent scenario, because (as we mentioned in the text) its perfect statistics can be reproduced with LV with a sufficiently large alphabet\footnote{It took some time for this obvious statement to be digested by those who had built their careers on proving that BB84 is ``unconditionally secure".}.

I had not been part of the discussions in reason of a momentary estrangement from Nicolas --- he had mentioned this research to me, but manifestly I had other things in mind. A few days after I came back in focus, Nicolas shared with me the first draft of what I realized would become a milestone paper \cite{agm}. I was enraged with myself, it was too late to do anything to add my name there: fortunately, I managed to give a positive turn to my rage and, in a few weeks, had produced most of the generalizations that appear in the full-length paper \cite{pra06}.

At the same time, Toni shared with me what he would consider as the next goal: to prove device-independent security against a quantum adversary. Indeed, security against no-signaling adversaries is conceptually appealing, but does not seem to be an urgent concern; moreover, it gives quite bad bounds. On the contrary, security against a ``normal" quantum adversary would be highly relevant for blind certification: nothing else than the idea Nicolas and I had discussed in Bristol, but reloaded in a fully new context. This time, I did not miss the chance, joined the effort from the start and ended up producing the core of the proof together with Serge Massar and Stefano Pironio \cite{acin07}. This is the first paper in which the wording ``device-independent" is actually used: if my memory does not betray me, the term must be attributed to Toni.

\subsection{Developments}

Our collaboration has had the good taste of leaving some assumptions in the security proofs: quite a number of prominent researchers worked for some years to close that gap. In the mean-time, I had had the idea of source certification and was introduced to the works of Mayers and Yao; eventually, Matthew McKague would crack the self-testing code by developing a simple version of the proofs, as we have seen. The idea of device-independent randomness certification, \textit{a priori} much simpler and more straightforward than cryptography, was brought forward only later: the merit must be shared between Roger Colbeck \cite{colbeck} and my previous co-authors Toni, Serge and Stefano \cite{randomnature}.

The device-independent program has also boosted the experimental quest for closing the detection loophole, something that was previously considered anecdotic at best (indeed, nobody had ever taken seriously a conspiracy of the detectors, but everyone should doubt the behavior of a black box). There are hints in the literature and in the corridors of conferences that the loophole-free Bell test is upcoming.

All these developments can be gathered from \cite{ourreview} or by browsing the arXiv.

\end{appendix}


\begin{thebibliography}{19}

\bibitem[Ac\'{\i}n, Gisin and Masanes 2006]{agm} Ac\'{\i}n, A., N. Gisin, L. Masanes, Phys. Rev. Lett. \textbf{97}, 120405 (2006)

\bibitem[Ac\'{\i}n \textit{et al.} 2007]{acin07} Ac\'{\i}n, A., N. Brunner, N. Gisin, S. Massar, S. Pironio, V. Scarani, Phys. Rev. Lett. \textbf{98}, 230501 (2007)

\bibitem[Bancal \textit{et al.} 2012]{bancal2012} Bancal, J.-D., S. Pironio, A. Ac\'{\i}n, Y.-C. Liang, V. Scarani, N. Gisin, Nature Physics \textbf{8}, 867 (2012)

\bibitem[Barrett, Hardy and Kent 2005]{bhk} Barrett, J., L. Hardy, A. Kent, Phys. Rev. Lett. \textbf{95}, 010503 (2005)

\bibitem[Barrett, Kent and Pironio 2006]{bkp06} Barrett, J., A. Kent, S. Pironio, Phys. Rev. Lett. \textbf{97}, 170409 (2006)

\bibitem[Bell 1964]{bell64} Bell, J.S., Physics \textbf{1}, 195 (1964)

\bibitem[Braunstein, Mann and Revzen 1992]{bmr92} Braunstein, S.L., A. Mann, M. Revzen, Phys. Rev. Lett. \textbf{68}, 3259 (1992)

\bibitem[Brukner \textit{et al.} 2004]{vienna} Brukner, C., M. Zukowski, J.-W. Pan, A. Zeilinger, Phys. Rev. Lett. \textbf{92}, 127901 (2004)

\bibitem[Brunner \textit{et al.} 2013]{ourreview} Brunner, N., D. Cavalcanti, S. Pironio, V. Scarani, S. Wehner, arXiv:1303.2849

\bibitem[Clauser, Horne, Shimony and Holt 1969]{chsh} Clauser, J.F., M.A. Horne, A. Shimony, and R.A. Holt, Phys. Rev. Lett. \textbf{23}, 880 (1969)

\bibitem[Colbeck and Kent 2011]{colbeck} Colbeck, R., A. Kent, J. Phys. A: Math. Theor. \textbf{44}, 095305 (2011)

\bibitem[Colbeck and Renner 2008]{cr08} Colbeck, R., R. Renner, Phys. Rev. Lett. \textbf{101}, 050403 (2008)

\bibitem[Degorre, Laplante and Roland 2005]{degorre} Degorre, J., S. Laplante, J. Roland, Phys. Rev. A \textbf{72}, 062314 (2005)

\bibitem[Einstein, Podolski and Rosen 1935]{epr} Einstein, A., B. Podolski, N. Rosen, Phys. Rev. \textbf{}, (1935)

\bibitem[Elitzur, Popescu and Rohrlich 1992]{epr2} Elitzur, A., S. Popescu, D. Rohrlich, Phys. Lett. A \textbf{162}, 25 (1992)

\bibitem[Fine 1982]{fine} Fine, A., Phys. Rev. Lett. \textbf{48}, 291 (1982)

\bibitem[Gill 2012]{gill2012} Gill, R.D., arXiv:1207.5103

\bibitem[Gisin 1991]{gisinthm} Gisin, N., Phys. Lett. A \textbf{154}, 201 (1991)

\bibitem[Gisin 2009]{gisinbellprize} Gisin, N., Science \textbf{326}, 1357 (2009)

\bibitem[Hall 2011]{Hall2011} Hall, M.J.W., Phys. Rev. A \textbf{84}, 022102 (2011)

\bibitem[Horodecki, Horodecki and Horodecki 1995]{horo3} Horodecki, R., M. Horodecki, P. Horodecki, Phys. Lett. A \textbf{200}, 340 (1995)

\bibitem[Liang and Doherty 2006]{liang} Liang, Y.-C., and A.C. Doherty, Phys. Rev. A, \textbf{73}, 052116 (2006)

\bibitem[Mayers and Yao 1998, 2004]{my04} Mayers, D., and A. Yao, Proceedings of the 39th IEEE
Conference on Foundations of Computer Science, 1998, page 503; available as quant-ph/9809039; \textit{idem}, Quant. Inf. Comput. \textbf{4}, 273 (2004)

\bibitem[McKague, Yang and Scarani 2012]{mys12} McKague, M., T.H. Yang, V. Scarani, J. Phys. A: Math. Theor. \textbf{45}, 455304 (2012)

\bibitem[Navascu\'es, Pironio and Ac\'{\i}n 2007-8]{npa} Navascu\'es, M., S. Pironio and A. Ac\'{\i}n, Phys. Rev. Lett. {\bf 98}, 010401 (2007); New J. Phys. {\bf 10}, 073013 (2008)

\bibitem[Navascu\'es and Wunderlich 2009]{nw09} Navascu\'es, M., H. Wunderlich, Proc. Roy. Soc. Lond. A \textbf{466}, 881 (2009) 

\bibitem[Norsen 2009]{norsen} Norsen, T., Found. Phys. \textbf{39}, 273 (2009)

\bibitem[Pawlowski \textit{et al.} 2009]{icnature} Pawlowski, M., T. Paterek, D. Kaszlikowski, V. Scarani, A. Winter, M. Zukowski, Nature \textbf{461}, 1101 (2009)

\bibitem[Pawlowski and Scarani 2013]{ps13} Pawlowski, M., V. Scarani, arXiv:1112.1142; to appear in: G. Chiribella, R. Spekkens (eds), \textit{Quantum theory: informational foundations and foils}, in preparation.


\bibitem[Pironio \textit{et al.} 2010]{randomnature} Pironio, S., A. Ac\'{\i}n, S. Massar, A. Boyer de la Giroday, D. N. Matsukevich, P. Maunz, S. Olmschenk, D. Hayes, L. Luo, T. A. Manning, C. Monroe, Nature \textbf{464}, 1021 (2010)

\bibitem[Popescu and Rohrlich 1992a]{prgisthm} Popescu, S., D. Rohrlich, Phys. Lett. A \textbf{166}, 293 (1992)

\bibitem[Popescu and Rohrlich 1992b]{pr92} Popescu, S., D. Rohrlich, Phys. Lett. A \textbf{169}, 411 (1992)

\bibitem[Popescu and Rohrlich 1994]{PR} Popescu, S., D. Rohrlich, Found. Phys. \textbf{24}, 379 (1994).

\bibitem[Preskill notes]{preskill} Preskill, J., lecture notes at Caltech, available at http://www.theory.caltech.edu/~preskill/ph229/

\bibitem[Scarani and Gisin 2005]{braz05} Scarani, V., and N. Gisin, Braz. J. Phys. \textbf{35}, 328 (2005)

\bibitem[Scarani 2006]{fff} Scarani, V., AIP Conference Proceedings, Vol. 844, pp. 309-320 (Melville, New York, 2006); available as arXiv:quant-ph/0603017

\bibitem[Scarani \textit{et al.} 2006]{pra06} Scarani, V., N. Gisin, N. Brunner, L. Masanes, S. Pino, A. Ac\'{\i}n, Phys. Rev. A \textbf{74}, 042339 (2006)

\bibitem[Toner and Bacon 2003]{tonerbacon} Toner, B.F., D. Bacon, Phys. Rev. Lett. \textbf{91}, 187904 (2003)

\bibitem[Tsirelson 1980]{tsi} Tsirelson, B.S., Lett. Math. Phys. \textbf{4}, 93 (1980)

\bibitem[Vazirani and Vidick 2011]{vv11} Vazirani, U., T. Vidick, Proceedings of the 44th Symposium on Theory of Computing, STOC'12, pages 61-76, ACM, 2011; available as arXiv:1111.6054

\bibitem[Vongehr 2012]{vongehr} Vongehr, S., arXiv:1207.5294


\bibitem[Werner 1989]{wer} Werner, R.F., Phys. Rev. A {\bf 40}, 4277 (1989)

\bibitem[Wood and Spekkens 2012]{spekkens} Wood, C.J., R.W. Spekkens, arXiv:1208.4119

\end{thebibliography}
\end{document}